\documentclass[superscriptaddress]{revtex4}

\usepackage{graphicx}
\usepackage{dcolumn}
\usepackage{bm}

\usepackage{amssymb}
\usepackage{amsmath}
\usepackage{amsthm}

\begin{document}


\def \Z {\mathbb{Z}}
\def \R {\mathbb{R}}
\def \C {\mathbb{C}}
\def \La {\Lambda}
\def \la {\lambda}
\def \ck {l}

\newcommand {\md} [1] {\mid\!#1\!\mid}
\newcommand {\be} {\begin{eqnarray}}
\newcommand {\ee} {\end{eqnarray}}
\newcommand {\ben} {\begin{eqnarray*}}
\newcommand {\een} {\end{eqnarray*}}
\newcommand {\tit} [1] {``#1''}

\newtheorem{theorem}{Theorem}
\newtheorem{lemma}{Lemma}
\newtheorem{corollary}{Corollary}

\title{Heat flow in anharmonic crystals with internal and external stochastic baths: A convergent polymer expansion for a model with discrete time and long range interparticle interaction}

\author{Emmanuel Pereira}
\email[Corresponding author: ]{emmanuel@fisica.ufmg.br}
\author{Mateus S. Mendon\c{c}a}
\email{mateussm@fisica.ufmg.br}
\affiliation{Departamento de F\'{\i}sica--Instituto de Ci\^encias Exatas, Universidade Federal de Minas Gerais, CP 702,
30.161-970 Belo Horizonte MG, Brazil}
\author{Humberto C. F. Lemos}
\email{humbertolemos@ufsj.edu.br}
\affiliation{Departamento de F\'{\i}sica e Matem\'atica, CAP -
Universidade Federal de S\~ao Jo\~ao del-Rei, 36.420-000, Ouro Branco, MG, Brazil}
\affiliation{Department of Physics, FMF, University of Ljubljana, Jadranska 19, 1000 Ljubljana, Slovenia}

%
%
%
%

%
%

%


%

\begin{abstract}
We investigate a chain of oscillators with anharmonic on-site potentials, with long range interparticle interactions, and coupled both to external and internal
stochastic thermal reservoirs of Ornstein-Uhlenbeck type. We develop an integral representation, a la Feynman-Kac, for the correlations and
the heat current. We assume the approximation of discrete times in the integral formalism (together
with a simplification in a subdominant part of the harmonic interaction) and develop a suitable polymer expansion for the model. In the regime of strong anharmonicity, strong harmonic pinning, and
for the interparticle interaction with integrable polynomial decay, we prove the convergence of the polymer expansion uniformly
in volume (number of sites and time). We also show that the two-point correlation decays in space
such as the interparticle interaction. The existence of a convergent polymer expansion is of practical interest: it
establishes a rigorous support for a perturbative analysis of the heat flow
problem and for the computation of the thermal conductivity in related anharmonic crystals, including those with inhomogeneous potentials and long range interparticle interactions. To show the usefulness and
trustworthiness of our approach, we compute the thermal conductivity of a specific anharmonic chain, and make a comparison with related numerical results presented in the literature.
\end{abstract}

\maketitle

\let\a=\alpha \let\b=\beta \let\c=\chi \let\d=\delta \let\e=\varepsilon
\let\f=\varphi \let\g=\gamma \let\h=\eta    \let\k=\kappa \let\l=\lambda
\let\m=\mu \let\n=\nu \let\o=\omega    \let\p=\pi \let\ph=\varphi
\let\r=\rho \let\s=\sigma \let\t=\tau \let\th=\vartheta
\let\y=\upsilon \let\x=\xi \let\z=\zeta
\let\D=\Delta \let\F=\Phi \let\G=\Gamma \let\L=\Lambda \let\Th=\Theta
\let\P=\Pi \let\Ps=\Psi \let\Si=\Sigma \let\X=\Xi
\let\Y=\Upsilon

\section{Introduction} \label{sec:intro}

The derivation of the macroscopic laws of heat transport from the underlying microscopic Hamiltonian models is still a
challenge in non-equilibrium statistical mechanics \cite{BLRB}. Since the pioneering work of Debye \cite{Deb} and Peierls \cite{Pei},
the   microscopic models recurrently used to describe heat flow in solids and crystals are mainly given by lattices of anharmonic
 oscillators, which lead to problems of considerable mathematical difficulty. Consequently,
 most of the works on the subject are carried out by means of computer simulations \cite{DharLLP,LiRMP}. There are, however,  some
 few mathematical results considering different aspects of the anharmonic heat flow problem: for example, there are rigorous works considering the existence of non-equilibrium
 stationary state \cite{EckmannPR,EckmannH,BLLO}; the rate of divergence of the thermal conductivity with the system size in Fermi-Pasta-Ulam type
 models \cite{LSpohn}; the on-set of Fourier's law in lattices with anharmonic on-site potentials \cite{Kup1, Kup2, LS}; the finiteness or infiniteness of the thermal conductivity given by a Green-Kubo  formula \cite{BLLO,BasileBOlla}.

 In systems with normal heat transport, the establishment of bounds showing the finiteness of the thermal conductivity is already an interesting and intricate problem. However, the
 derivation of more precise expressions, which seems to be an exceedingly difficult task,
 is highly desirable, both for fundamental reasons as well as in order to provide useful information about properties of
 the heat conduction with experimental applications, such as the possibility of thermal rectification, the existence of negative differential thermal resistance, etc. In particular, concerning possible applications,
 it is worth to recall the progress
 of Phononics \cite{LiRMP}, the counterpart of electronics devoted to the manipulation and control of the heat current. A
  considerable effort has been dedicated to the development of Phononic devices, systems idealized to work as electronic analogs, such as thermal diodes, thermal
 transistors, etc. The basic phenomenon behind the operation of these devices is the thermal rectification, which means asymmetric heat flow, and its understanding involves the investigation of
 inhomogeneous, asymmetric materials. There is an intense research in this subject, but, again, most of the results are obtained by means of computer simulations or numerical techniques, and so, a profitable analytical approach is
 opportune.

 In the present work, we aim to develop an approach that allows the detailed study of the heat flow and that can be used in the computation of an expression for the thermal conductivity in reasonable approximations of recurrent
 lattice models of anharmonic oscillators. Namely, we consider the chain of oscillators with anharmonic on-site potentials, nonlocal interparticle interactions and with stochastic baths coupled to each site: a model in which the
 time evolution is given by a combination of deterministic and stochastic dynamics.
  We develop an integral representation (a la Feynman-Kac) for the correlation functions, which are related to the heat flow and
 to the thermal conductivity. To make treatable the analysis, we introduce a discrete time regularization, i.e., an ultraviolet cutoff. To be free of huge subdominant terms and avoid unimportant technical difficulties, we also consider a simplified expression for the subdominant harmonic interaction. For this  discrete time and simplified version, we present a suitable polymer (cluster) expansion, and
 prove its convergence uniformly in volume (arbitrary number of sites and arbitrary times). We consider a system with interparticle interactions beyond nearest neighbor sites: precisely, our approach
 is also valid  for interparticle potentials with integrable polynomial decay. For these systems with long range interactions, we show that the two-point correlation function,
 which is directly related to the heat current and the thermal conductivity, decays in space such as the interparticle interaction.

We must emphasize that the existence of an integral representation for the correlations and a convergent cluster expansion for such representation, the main technical achievement of the present work,  is of practical importance: it allows an accurate perturbative investigation of the heat flow problem in these systems with anharmonic oscillators. Precisely, the convergence
of the expansion proves that the terms with small polymers and with small sizes, which correspond to the terms of lower order in a ``naive'' perturbative expansion for the potential interaction, contain already the main information about the model.
That is, a convergent polymer expansion
appears as a support for the validity of some theoretical results previously obtained by means of non rigorous perturbative computations \cite{PF,PF2,PA,PFL,Prapid}, and may
 provide a precise tool for further investigation of the heat flow problem, even in more intricate systems such as inhomogeneous, graded chains or
models with long range interparticle interactions. In specific, for these much more complicate models with long range interactions, we recall that present technology permits the fabrication of such systems. For example, nowadays nanomagnets of Permalloy are lithographically manipulated to present interesting properties \cite{Perma}: these materials are intensively investigated and their interparticle interactions present polynomial decay (such as that considered in the present work).

We stress that the model to be investigated here, given by a chain of anharmonic oscillators (with inner noises, representing extra effective interactions missing in the Hamiltonian), is a natural model for the investigation of the heat conduction in the nonequilibrium steady state of electric insulating solids submitted to different temperatures.
As a motivation to understand the importance and usefulness of our approach and results for the (necessary and difficult) analytical investigation of such anharmonic systems,
and also as a justification for the discrete approximation to be used here, we recall some recent and important related works carried out within considerable approximations. A very elaborate work involving
lattices of anharmonic oscillators  is that due to Bricmont and Kupiainen \cite{Kup1, Kup2}. In these articles, restricted to systems with space dimension $d\geq 3$, the authors take a system with baths at the boundaries only and derive the Fourier's law  for the case of a quartic on-site potential. As well known, Fourier's law is the phenomenological basic law for the heat transport which states that the
heat flow is proportional to temperature gradient. They show that the correlation functions of the system satisfy an infinite set of linear
equations (Hopf equations), and, to carry out the investigation,  they have to assume a closure approximation to these equations (according to the authors: ``this is an uncontrolled approximation that we do not know how to justify rigorously'' \cite{Kup2}).
To emphasize, again, the difficulty of the analytical study, we quote some other comments:
``despite its fundamental nature, a derivation of Fourier's law from
first principles lies well beyond what can be mathematically
proven'' \cite{Kup1}; ``a first principle derivation of the law is
missing and, many would say, is not even on the horizon''
\cite{Kup2}. Another elaborate and interesting related work is that due to Olla and collaborators \cite{BasileBOlla}. There, to investigate the relation between normal heat transport and space dimension in systems with
momentum conservation, the authors consider a hypothetical mathematical model given by harmonic oscillators but perturbed by a nonlinear stochastic dynamics conserving momentum and energy. And the authors say:
``a rigorous treatment of a non-linear system, even the
proof of the conductivity coefficient, is out of reach of current
mathematical techniques'' \cite{BasileBOlla}.

Given such scenario, in this present paper, in order to perform rigorous analytical investigations in this basic and recurrent model given by anharmonic lattice of oscillators, we also have to assume an approximation: namely, after
establishing a rigorous integral representation for the correlations (related to the heat flow), we simplify the dynamics by assuming the evolution given by discrete times. Within such time discretization, we are able to
rewrite the integral representation, which originally involves a Gaussian measure, in terms of a new intricate measure with anharmonic terms. We stress that it is a central point: by starting with this new ``correct'' measure (which really involves
the anharmonic interaction part), we are able to develop a convergent polymer expansion, or, in other words, we can perform a rigorous (convergent) perturbative analysis in the problem, which seems impossible with the original
Gaussian measure. Of course,  the exceeding difficulty of the original problem remains: we do not know how to control the limit of time discretization going to zero, i.e., we cannot treat the original problem with continuous time. This complication in recovering the continuous limit from a discrete version is very common in physical problems with intricate interactions: recall, for example, the very arduous study of the
ultraviolet limit of models in Quantum Field Theory (QFT) with a cutoff \cite{renormalization}. Here, we control the infrared limit of the model: i.e.,  we develop the polymer expansion and
obtain uniform bounds, which are valid even if the space and time volume goes to infinite.
Finally, we must emphasize that our formalism with discrete times gives a very accurate result in a comparison with well known numerical works. Precisely, when applied to the computation of the thermal conductivity dependence on temperature for the anharmonic chain with quartic on-site potential. That is, we believe to have a good approach for this difficult problem of anharmonic oscillators. Moreover, as already said, our formalism is also extended to intricate systems with long range interactions, systems with important physical properties.

We organize this paper as follows.
In section \ref{sec:model}, we introduce the model and develop an integral representation for the correlations, that are directly related to the heat flow.
In this integral formalism, we still introduce the time discretization and a simplification for part of the harmonic interaction (simplification that is discussed in appendix \ref{apendice}).
In section \ref{sec:poly}, we introduce a polymer
expansion for the model. We prove the convergence of such polymer expansion in section \ref{sec:conv}.
In section \ref{sec:decay}, we show the convergence for the modified polymer expansion associated to the two-point correlation. A concrete example, which makes clear the direct application of our results in the detailed computation of the thermal conductivity, is presented in section \ref{sec:exemplo}. Section \ref{sec:final} is devoted for final remarks.

\section{The Model and The Integral Representation} \label{sec:model}

We describe our anharmonic crystal model. We consider a lattice system with unbounded variables, coupled to both external and internal heat baths of Ornstein-Uhlenbeck type. More precisely, we take a
system of $N$ oscillators with Hamiltonian
\begin{equation}
H(q,p)\! = \!\!\! \sum_{j=1}^{N}\! \left[ \frac{1}{2} \! \left({p_j^2} + M_j
q_j^2  + \sum_{l\neq j}q_l J_{lj}q_j \right) + \lambda P(q_{j})
\right] , \label{Hamiltonian}
\end{equation}
where $q$ and $p$ are vectors in $\mathbb{R}^{N}$; $\lambda,  M_j >0$; $J_{jl}=J_{lj}=f(|\ell-j|)$, $f$ with some integrable decay (details ahead);  ${P}$ is the anharmonic on-site potential, which we take as the polynomial ${P}(q_j) = q_{j}^{4}/4$. For
simplicity, we take here the particle masses as 1 (but our method follows also for general cases, including inhomogeneous distributions for the
particle masses), and we assume the space dimension $d=1$. We take the dynamics given by the stochastic differential
equations
\begin{equation}
dq_j\! =\! {p_j} dt , ~~
dp_j\! =\! -\frac{\partial H}{\partial q_j}dt-\zeta_{j}
p_jdt+\gamma^{1/2}_j dB_j  , \label{eqdynamics}
\end{equation}
where $B_j$ are independent Brownian motions, with zero average and diffusion equal to 1
\begin{equation} \label{noise}
\left< B_{j}(t)\right> = 0, ~~ \left< B_{j}(t)B_{\ell}(s)\right> = \delta_{i,j}{\rm min}(t,s) ,
\end{equation}
 $\zeta_{j}$ is the constant coupling between site
$j$ and its reservoir; $\gamma_j=2\zeta_j T_j$, where
$T_j$ is the temperature of the $j$-th  bath.

To study the heat flow, we first define the energy of a single oscillator
\begin{equation}
H_{j}(q_{j},p_{j})\! = \!\ \frac{p_j^2}{2}   + \frac{1}{2}\sum_{l\neq j}V(q_l - q_j) + V_{2}(q_j)  , \label{Hamiltoniansingle}
\end{equation}
where $H(q,p) = \sum_{j}H_{j}(q_{j},p_{j})$. The expression for $V$ comes after writing the interparticle potential in the  Hamiltonian above
 as $\frac{1}{2}\sum_{\ell\neq j}V(q_{j}-q_{\ell}) =$
$ \frac{1}{2}\sum_{\ell\neq j} J_{j\ell}(q_{j}-q_{\ell})^{2}$
(with adjustments in $M_j$, the pinning constant, coefficient of $q_{j}^{2}$);
$V_{2}$ describes the on-site potentials above and involves the terms $\lambda P(q_{j}) + M_{j}q_{j}^{2}$ plus some terms with $q_{j}^{2}$, which
appear as we write $q_{j}J_{j,k}q_{k}$ as $J_{j,k}q_{j}q_{k}/2$ (as said, these terms may be treated as an adjustment in $M_{j}$). Thus, we have
\begin{eqnarray}
& &\left<\frac{dH_j}{dt}(t) \right> =  \mathcal{R}_j(t) + \left< \mathfrak{F}_{\rightarrow j}-\mathfrak{F}_{j\rightarrow} \right>,
\\
& &\mathfrak{F}_{j\rightarrow}\!\!=\!\! \sum_{\ell>j}\nabla_{j} V(q_j-q_{\ell})\left(\frac{p_j}{2}+\frac{p_{\ell}}{2}\right)
 = \sum_{\ell>j}J_{j\ell}(q_j-q_{\ell})\left(\frac{p_j}{2}+\frac{p_{\ell}}{2}\right)\!\!,\label{fluxo1}
\\
& &\mathfrak{F}_{\rightarrow j}\!\!=\!\!\sum_{\ell<j}\nabla_{j} V(q_j-q_{\ell})\left(\frac{p_j}{2}+\frac{p_{\ell}}{2}\right)
 = \sum_{\ell<j}J_{j\ell}(q_j-q_{\ell})\left(\frac{p_j}{2}+\frac{p_{\ell}}{2}\right)\!\!, \label{fluxo2}
\\
& &\mathcal{R}_j(t) = \zeta_{j}\left( T_{j} - \left< p_{j}^{2}\right>\right) \label{reservatorio}.
\end{eqnarray}
(More details about the derivation of such equations may be found, e.g., in Refs.\cite{PF, BLL}.)
$\mathfrak{F}_{j\rightarrow}$ gives the heat current from site $j$ to the forward sites $\ell>j$; $\mathfrak{F}_{\rightarrow j}$ gives the current
from the previous sites $\ell<j$. $\mathcal{R}_{j}$ denotes the energy flux between the $j$-th site and the $j$-th reservoir. These models with internal
stochastic reservoirs are recurrent, and have been considered in several works \cite{BLLO,BRV,BLL,PF}, usually with the self-consistent
condition, which means that the temperatures of the internal reservoirs are chosen such that there is no net energy flux between these internal baths
and the system in the steady state, i.e., $\lim_{t\rightarrow\infty}\mathcal{R}_{j}(t) = 0$, for $j= 2, 3, \ldots, N-1$. In other words, in the stationary state with the self-consistent condition we get a heat current across the system supplied only by the external
baths at the boundaries with different temperatures. The existence of a steady state in the system with the self-consistent condition (that is, the existence of
this suitable choice of internal temperatures) is proven in Ref.\cite{SZ} and Ref.\cite{BLLO}, for the harmonic and anharmonic cases, respectively. In the present paper, we assume that
the temperatures $T_{1}, T_{2}, \ldots, T_{N-1},T_{N}$ are arbitrarily given, chosen from a set with lower bound, i.e., there is a $T_{min}$ such that $T_{min} \leq T_{1}, T_{2}, \ldots,
T_{N}$, for all $N$. In the Final Remarks section, we recall some physical problems that consider the self-consistent condition.

It is interesting to note the generality of the temperature distribution allowed here. The self-consistent condition, usually assumed with these models with inner reservoirs (and which is considered in
the example described in section \ref{sec:exemplo}), is related to a specific temperature
profile as recalled above, but our approach follows also for many other cases. An interesting problem involving such models of oscillators with inner reservoirs but without the self-consistent condition is presented in Ref.\cite{Zabey}.

To proceed, we introduce the
phase-space vector $\varphi=(q,p) \in \mathbb{R}^{2N}$, and write the dynamics (\ref{eqdynamics}) as
\begin{equation}
d{\varphi}= -A \varphi dt -\lambda{P}'(\varphi) dt + \sigma dB,
\label{dynamics}
\end{equation}
where $A=A^0+\mathcal{J}$ and $\sigma$ are $2N\times 2N$ matrices
\begin{eqnarray}
A^0=\left ( \begin{array}{cc} 0 & - I \\
\mathcal{M} & \Gamma  \end{array} \right ),~~~~\mathcal{J}=\left (
\begin{array}{cc} 0 & 0 \\
\mathbb{J} & 0 \end{array} \right ),~~~~\sigma=\left (\begin{array}{cc} 0 & 0 \\
0 & \sqrt{2\Gamma\mathcal{T}} \end{array}\right ),
\end{eqnarray}
$I$ is the unit $N\times N$ matrix; $\mathbb{J}$ is the $N\times N$
matrix for the interparticle interactions; $\mathcal{M}$,
$\Gamma$ and $\mathcal{T}$ are the diagonal $N\times N$ matrices, with positive elements:
$\mathcal{M}_{jl}=M\delta_{jl}$, $\Gamma_{jl}=\zeta\delta_{jl}$,
$\mathcal{T}=T_j\delta_{jl}$. $B$ are independent Brownian motions; ${P}'(\varphi)$ is a $2N \times 1$ matrix with
${P}'(\varphi)_j=0$ for $ j=1,\ldots,N$ and
\begin{equation}
{P}'(\varphi)_i=\frac{d{P}(\varphi_{i-N})}{d\varphi_{i-N}}
\quad  \mbox{for}  \quad i=N+1,\ldots,2N.
\end{equation}

For while,  let us use the following index notation: $i$ for indices
in the set $\{N+1,N+2,\cdots,2N\}$ (related to momenta coordinates); $j$ for values in the set
$\{1,2,\cdots,N\}$ (related to space position coordinates),  and $k$ for values in $\{1,2,\cdots ,2N\}$. Throughout the paper, we will
 omit obvious sums over repeated indices.

As described above (\ref{fluxo1}, \ref{fluxo2}), the heat flux across the chain is given in terms of two-point functions. Thus, to
obtain a mechanism to study heat conduction, we develop an integral representation for the correlation
functions, in which a rigorous control is possible (after adjustments such as time discretization) by using standard methods of field theory and equilibrium
statistical physics, namely, polymer expansions.

We start the construction of the integral formalism with the solution of the linear ($\lambda = 0$) and decoupled
($\mathcal{J}\equiv 0$) dynamical system. We have

\begin{lemma} \label{lemma:OU}
The solution $\phi(t)$ of Eq.(\ref{dynamics}) with $\lambda=0, \mathcal{J}\equiv 0 $, i.e. of equation
\begin{equation*}
d{\phi}=-A^{0} \phi dt + \sigma dB,
\end{equation*}
is the
Ornstein-Uhlenbeck Gaussian process
\begin{equation}
\phi(t)=e^{-tA^0}\phi(0)+\int_0^t e^{-(t-s)A^0}\sigma dB(s),
\end{equation}
where, for the case of $\phi(0)=0$, the covariance of the
process evolves as
\begin{eqnarray}
\left <\phi(t)\phi(s)\right >_0\equiv \mathcal{C}(t,s)=\left \{
\begin{array}{c} e^{-(t-s)A^0}\mathcal{C}(s,s), ~~t\geq s, \label{covariance1} \\
\mathcal{C}(t,t)e^{-(s-t)A^{0^\dagger}}, ~~t\leq s,\end{array}\right . \\
\mathcal{C}(t,t)=\int_0^tdse^{-sA^0}\sigma^2e^{-sA^{0^\dagger}}.
\label{covariance}
\end{eqnarray}
\end{lemma}

\begin{proof}
Exercise of stochastic
differential equations (see e.g. Ref.\cite{Ok}).
\end{proof}

We recall that these solutions may be realized as continues trajectories. Moreover, it follows that
\begin{lemma}
\begin{eqnarray}
\lim_{t\rightarrow\infty}\mathcal{C}(t,t) \equiv C = \int_0^{\infty}ds e^{-sA^0}\sigma^2e^{-sA^{0^\dagger}}=\left
(\begin{array}{cc}\frac{\mathcal{T}}{\mathcal{M}} &0 \\
0&\mathcal{T}\end{array} \right),\label{ecovariance}
\end{eqnarray}
and, for any $\alpha$ such that $0 < \alpha < {\rm min}\left\{\frac{\zeta}{2},\frac{M}{\zeta}\right\}$, there is a
constant $c < \infty$ such that, for all $t>0$ and for all $N$,
\begin{equation}
\parallel e^{-tA^{0}}\parallel \leq c e^{-t\alpha}. \label{cota}
\end{equation}
\end{lemma}

\begin{proof}
See Ref.\cite{SZ} for the proof of Eq.(\ref{ecovariance}) and Ref.\cite{BLL} for Eq.(\ref{cota}).
\end{proof}

It is worth to remark that, for the harmonic and decoupled system (i.e., with $\lambda = \mathcal{J} = 0$), each site $j$ is isolated from the other
sites and coupled to a single bath at temperature $T_{j}$. Then, the expected steady distribution (as $t \rightarrow \infty$) is the related Boltzmann-Gibbs
measure, given by the Gaussian measure with measure $d\mu_{C}$, with the convariance described by Eq.(\ref{ecovariance}), in which each site has the temperature $T_{j}$ of the bath coupled to it.

Now we use the Girsanov
theorem to describe a representation for the correlation functions of $\varphi(t)$, the solution for the complete process (\ref{dynamics}). We will construct an integral representation for a system
with $N$ sites and with time running from $0$ to $\mathfrak{T}$. Later, after the time discretization, we will obtain bounds, valid for all $N$ and $\mathfrak{T}$, leading to the convergence of the associated polymer expansion.

\begin{theorem}
For the correlation functions (\ref{fluxo1}-\ref{reservatorio}) of the crystal chain with reservoirs
at each site (\ref{Hamiltonian}-\ref{noise}), we have the integral representation given by
\begin{equation}
\label{correlations}
\left<\varphi_{\ell_{1}}(t_{1})\ldots \varphi_{\ell_{k}}(t_{k})\right> =
\int \phi_{\ell_{1}}(t_{1})\ldots\phi_{\ell_{k}}(t_{k})\exp[-W(\phi)]
d\mu_{\mathcal{C}}, ~~~~t_{1},\ldots, t_{k}\leq \mathfrak{T},
\end{equation}
with
\begin{eqnarray*}
W(\phi) &=& \int_{0}^{\mathfrak{T}} \!\!\!\! \phi_j(s)\mathcal{J}^{\dagger}_{ji}\gamma_i^{-1}d\phi_i(s) +
\lambda\gamma_i^{-1}{P}'(\phi)_i(t)d\phi_i(s) +
 \phi_j(s)\mathcal{J}_{ij}^{\dagger}\gamma_{i}^{-1}A^{0}_{ik}\phi_k(s) ds + \\
&& + \lambda
\gamma_{i}^{-1} {P}' (\phi)_i(s) A^0_{ik} \phi_k(s) ds +
\frac{1}{2}\phi_{j'}(s)\mathcal{J}_{j'i}^{\dagger}\gamma_i^{-1}\mathcal{J}_{ij}\phi_j(s) ds +
\\
&& + \frac{1}{2}\lambda^2\gamma_i^{-1}({P}'(\phi)_i)^2(s) ds + \lambda \gamma_{i}^{-1}
{P}' (\phi)_i(s) \mathcal{J}_{ij}\phi_j(s) ds,
\end{eqnarray*}
where $\phi$ is the solution (given by lemma \ref{lemma:OU}) of the process with
$\mathcal{J}\equiv0$, $\lambda=0$, and $\varphi$ is the solution for
the complete process (\ref{dynamics}); the covariance $\mathcal{C}$
is given by equations (\ref{covariance1}) and (\ref{covariance}). The sum over repeated indices $i, j, k, \ldots$ is assumed above (and throughout the paper, as already said).
\end{theorem}

\begin{proof}
Girsanov theorem (see e.g. theorem $8.6.8$ in Ref.\cite{Ok}; see also Ref.\cite{BS}) gives a measure $\mu_{*}$ for the new process $\varphi$ in
terms of the measure $\mu_{\mathcal{C}}$ associated to previous Ornstein-Uhlenbeck process $\phi$ with $\mathcal{J}\equiv 0$ and $\lambda = 0$.
Precisely, for any measurable set $A$, it follows that $\mu_{*}(A) = E_{0}(1_{A}Z(\mathfrak{T}))$, where $E_{0}$ is the expectation of $\mu_{\mathcal{C}}$,
$1_{A}$ is the characteristic function, and
\begin{eqnarray}
Z(\mathfrak{T})=\exp\left(\int_{0}^{\mathfrak{T}}u\cdot dB
-\frac{1}{2}\int_{0}^{\mathfrak{T}}u^2ds\right), \label{girs} \quad \mbox{where }
\gamma_i^{1/2}u_i=-\mathcal{J}_
{ik}\phi_k-\lambda{P}'(\phi)_i .\nonumber
\end{eqnarray}
The inner products above are in $\mathbb{R}^{2N}$. Note that, following our previous index convention, $u_{i}$ is nonzero only for $i \in \{N+1, N+2, \ldots, 2N\}$.
To apply the Girsanov theorem we must show that $Z(t)$ is a martingale with respect to the $\sigma$-algebra generated by $\phi(t)$ and $\mu_{\mathcal{C}}$. To show it,
we define the It\^o process
$$
dX(t) = -\frac{1}{2}u^{2} dt + u\cdot dB, ~~~~ X(0)=0,
$$
with $u$ as previously defined. Then, $Z(t) = \exp[X(t)]$ is also an
It\^o process and
$$
dZ(t) = Z(t)u\cdot dB(t) \Rightarrow Z(t) = 1 +
\int_{0}^{t}Z(s)u(s)\cdot dB(s).
$$
As $\phi$ admits a
continuum realization, it follows that $Z(t)$ is bounded  and  $uZ$ is square-integrable, i.e.,
$$
E_{0}\left(\int_{0}^{t} u^{2}(s)Z^{2}(s) ds\right) < \infty.
$$
And so, it follows that $Z(t)$ is a martingale (see e.g. corollary $3.2.6$ in Ref.\cite{Ok}).

To conclude, we note that
\begin{eqnarray*}
u_idB_i=\gamma_i^{-1/2}u_i\cdot\gamma_i^{1/2}dB_i&=&\gamma_i^{-1/2}u_i\cdot\left
(d\phi_i+A^0_{ik}\phi_kdt \right )
=-\gamma_i^{-1}\left(\mathcal{J}_{ij}\phi_j+\lambda{P}'(\phi)_i\right
) \left (d\phi_i+A^0_{ik}\phi_k dt\right ) ,
\end{eqnarray*}
and, finally, we write $u^{2}$ in terms of $\phi$, see Eq.(\ref{girs}), to obtain the expression as claimed in the theorem.
\end{proof}

The detailed study of the integral representation above,
given by a complicate anharmonic perturbation  of a Gaussian measure,  seems to
be exceedingly difficult. And so, to carry out the rigorous investigation, we try to rewrite the integral representation as an expression in which a more appropriate measure may be considered. But, to establish such new representation, some simplification is necessary. In other words, we propose to follow the investigation in
an approximated version of the original problem, derived from the previous formalism as follows (version in which, such strategy of considering a
suitable non Gaussian measure is possible).

First, we make an important modification: we introduce an ultraviolet cutoff in the time integral; precisely, we assume discrete times $t= \varepsilon, 2\varepsilon, \ldots, \mathfrak{T}$.
Second, to avoid
unimportant technical difficulties and  huge expressions for some subdominant terms (easily controlled in the forthcoming polymer expansion), we simplify the expression
for the covariance $\mathcal{C}$ associated to the previous Gaussian process $\phi$, given by Eqs.(\ref{covariance1}, \ref{covariance}). Precisely, we  replace  the Gaussian measure in the integral representation with such covariance $\mathcal{C}$ by its main part
\begin{eqnarray} \label{quadratica}
&d\mu_{\mathcal{C}}& \longrightarrow \left.\exp\left(-\frac{1}{2}\sum_{k,k',t,t'}\phi_{k}(t)\varepsilon \mathcal{D}^{-1}(t-t')\phi_{k'}(t')\right) \prod_{k,t} d\phi_{k}(t)\right/\mathcal{N},
\\
&\mathcal{D}^{-1}& \equiv C^{-1} \left(\frac{M}{\zeta}\delta_{t,t'} + c_{1}[-\Delta(t,t')]\right)  ~,
\nonumber
\end{eqnarray}
where $\mathcal{N}$ is the normalization; $-\Delta$ is the discrete Laplacian, $-\Delta(t,s) = 2\delta_{t,s} - \delta_{|t-s|,1}$; and $c_{1} = \mathcal{O}(1/\alpha)$ is a small parameter: we assume large $\alpha \equiv \zeta/2$ (i.e., large dissipation), and still take
a strong pinning $M = 3\alpha^{2}$. Details in  Appendix A, where  we show that this new covariance describes, indeed, the main part of the original Gaussian measure.

From the study of  chains of oscillators, it is well known  that Fourier's law holds in the harmonic system with
internal self-consistent reservoirs \cite{BLL}, but it does not hold anymore if these internal reservoirs are turned off \cite{RLL}. The scenario is
different for the chain with anharmonic on-site potentials: one expects that Fourier's law will be obeyed even without the internal baths. As it is
away from our present skill to prove it, to proceed we ignore the possibility of different coupling constants with the internal or external reservoirs and
take the same $\zeta_{j}= \zeta$ for all sites.

Hence, with discrete times and with the simplification of the harmonic covariance, the representation for the correlations
(\ref{correlations}), e.g. for the two-point function, becomes
\begin{eqnarray*}
\lefteqn{\left<\varphi_{\ell_{1}}(t) \varphi_{\ell_{2}}(t)\right> \simeq
\int  \phi_{\ell_{1}}(t)\phi_{\ell_{2}}(t)
 \exp\left\{ -\sum_{s, i,j,\ldots} \varepsilon\left[
  \phi_{j}(s)\mathcal{J}_{ji}^{\dagger}\gamma_{i}^{-1}\frac{[\phi_{i}(s+\varepsilon)-\phi_{i}(s)]}{\varepsilon} + \right.\right.} \\
   && + \gamma_{i}^{-1}\lambda
{P}'(\phi_{i-N}(s)) \frac{[\phi_{i}(s+\varepsilon)-\phi_{i}(s)]}{\varepsilon}
+ ~ \phi_{j}(s) \mathcal{J}_{ji}^{\dagger}\gamma_{i}^{-1}[M_{i-N}\phi_{i-N}(s) + \zeta\phi_{i}(s)] + \\
&&   + \gamma_{i}^{-1}\lambda_{i}
{P}'(\phi_{i-N}(s)) M_{i-N}[\phi_{i-N}(s) + \zeta\phi_{i}(s)]
 + \frac{1}{2}\phi_{j'}(s) \mathcal{J}_{j'i}^{\dagger}\gamma_{i}^{-1}\mathcal{J}_{ij}\phi_{j}(s) + \\
 &&   + \frac{1}{2}\gamma_{i}^{-1}\lambda^{2}
[{P}'(\phi_{i-N}(s))]^{2} + \gamma_{i}^{-1}\lambda {P}'(\phi_{i-N}(s))\mathcal{J}_{ij}\phi_{j}(s) + \\
&& + \left.\left.\left.
\frac{1}{2}\phi_{k}(s)
\mathcal{D}_{k,k'}^{-1}(s,s') \phi_{k'}(s')
\right] \right\}  \prod_{s,k}d\phi_{k}(s) \right/ {\rm normalization} ;
\end{eqnarray*}
moreover, taking $\varepsilon = 1/\zeta$ for simplification, we obtain
\begin{eqnarray}
\label{principal}
& & \left<\varphi_{\ell_{1}}(t_{1}) \varphi_{\ell_{2}}(t_{2})\right> \simeq
\int  \phi_{\ell_{1}}(t_{1})\phi_{\ell_{2}}(t_{2})
 \exp\left\{ -\sum_{s, i,j,\ldots} \varepsilon\left[
  \phi_{j}(s)\mathcal{J}_{ji}^{\dagger}\gamma_{i}^{-1}\phi_{i}(s+\varepsilon) + \right.\right.
\nonumber \\
&& + \frac{\lambda}{\gamma_{i}} {P}'(\phi_{i-N}(s)) \phi_{i}(s+\varepsilon)  +
\phi_{j}(s) \mathcal{J}_{ji}^{\dagger}\frac{M_{i-N}}{\gamma_{i}}\phi_{i-N}(s) +
\nonumber \\
&& + \!\! \frac{\lambda_{i}}{\gamma_{i}} {P}'(\phi_{i-N}(s)) M_{i-N}\phi_{i-N}(s)\! +\!
\frac{1}{2\gamma_{i}}\phi_{j'}(s) \mathcal{J}_{j'i}^{\dagger}\mathcal{J}_{ij}\phi_{j}(s)\! +\!
\frac{\lambda^{2}}{2\gamma_{i}} [{P}'(\phi_{i-N}(s))]^{2} +
\nonumber \\
&& + \left. \left. \frac{\lambda}{\gamma_{i}} {P}'(\phi_{i-N}(s))\mathcal{J}_{ij}\phi_{j}(s) +
\frac{1}{2}\phi_{k}(s) \mathcal{D}_{k,k'}^{-1}(s,s') \phi_{k'}(s')
\right] \Bigg\}\!\! \prod_{s,k}d\phi_{k}(s)\!\!\! \right/ \widetilde{\mathcal{N}} ,
\end{eqnarray}
where, as previously established, $s, s', t, \ldots \in \{\varepsilon, 2\varepsilon, \ldots, \mathfrak{T}\}$; $j,j' \in \{1,2,\ldots,N\}$; $i,i' \in \{N, N+1, \ldots, 2N\}$; $k,k' \in \{1,2, \dots, 2N\}$;   the denominator $\widetilde{\mathcal{N}}$ above is the numerator with $\phi_{\ell_{1}}=\phi_{\ell_{2}} = 1$, and it was introduced to keep normalized the measure $\exp\{\ldots\}\prod d\phi_{k}(s)$: the original measure given by $Z(\tau)d\mu_{\mathcal{C}}$, which appears before
the changes due to time discretization and the simplification of the harmonic part,  is normalized.

The main technical achievement of the integral discrete time representation above is that now, as we aimed, a mathematical investigation starting from suitable measures (non Gaussian distributions) will be possible,
allowing a profitable perturbative analysis. The polymer expansion, described in the
next section, makes it clear.



\section{The polymer expansion} \label{sec:poly}

We now develop a specific polymer expansion \cite{Si,GJ,PS}, suitable for our problem. The technique of polymer (or cluster) expansion is mathematically involved, but it recurrently used in different areas of physics,
such as phase transitions in equilibrium statistical mechanics, spectral analysis in field theory, etc. Detailed reviews may be found in textbooks such as Refs.\cite{Si, GJ, B}.
The existence of a convergent polymer expansion for our present system will allow us
to obtain the precise decay of the correlation functions, that is of crucial importance for the study of heat flow in a system with interparticle interactions
beyond nearest-neighbor sites \cite{PA}. More importantly, a convergent polymer expansion establishes a rigorous support for a perturbative approach for the heat flow investigation, as repeatedly emphasized.

To start the polymer expansion, we need to properly rewrite, reorganize the integral formalism. And so, before describing the technical structures, let us stress the reason of such rearrangement: in these problems involving systems with anharmonic interactions whose behavior is quite different
from that observed in related system but with harmonic interactions only (as we have in the heat flow problem for chains of oscillators), a perturbative analysis of the anharmonic model within an integral representation starting with a Gaussian measure (which comes from the harmonic interaction part only) is doomed to failure. For such reason, we are forced to rewrite the integral representation for the correlations in terms of non Gaussian measures: in
our case, suitable anharmonic single spin distributions to be described ahead. And so, a perturbative analysis  makes sense now: that is, within an integral formalism involving a measure with enough information about the anharmonic potential, we have a suitable starting point, and so, the complete result is reached by adding (now) small corrections to this anharmonic part, corrections which generate
a convergent perturbative series.

The notation to be used here and in the following sections is somehow intricate, but is usual in works involving polymer expansion theory, see e.g.
Ref.\cite{PS} and references there in.

We first consider the term which we name as partition function $Z_{\Lambda}$ (as usually named in theory of polymers), that is, the denominator of the two-point function above (\ref{principal}) (i.e. the numerator with $\phi_{\ell_{1}}=\phi_{\ell_{2}} = 1$),
and rewrite it in terms of polymers. But, instead of an usual single spin distribution,
we take as local distributions the measures associated to ``cells'' of $\psi_{x}$, where  $\psi_{x} = (q_{x}, p_{x})$, with
$$q_{x} = \lambda^{1/3}\phi_{\vec{x}}(x_{0}) ~~~~ {\rm and} ~~~~ p_{x} = \phi_{\vec{x} + N}(x_{0}+\varepsilon),$$
$(x_{0},\vec{x}) \in \Lambda = \left\{\varepsilon, 2\varepsilon, \ldots, \mathfrak{T}\right\}\times\left\{1, 2, \ldots, N\right\} \subset \mathbb{Z}^{2}_{*} \equiv \varepsilon\mathbb{Z}\times\mathbb{Z}$.
We remark that, for clearness in the forthcoming manipulation with polymers,  a new notation was introduced above: we replaced
 the previous time and index notations $t, i, j, \ldots$ for $(x_{0},\vec{x}) \equiv x$; $x_{0}$ for time, and $\vec{x}$ for space. When necessary, we will also split the parts of $\psi$ as $q$ and $p$. Note that our basic ``cell'' to be used in the polymer expansion involves $\psi_{x}$, in which $q$ and $p$ are in the same site (i.e., in the same space position $\vec{x}$), but they are nearest neighbors in time.

As the local measure, we define (for $x_{0} \neq \varepsilon$ or $\mathfrak{T}$)
\begin{equation}
\label{dnu}
d\n(\psi_{x})= {{e^{- U(\psi_{x})} \over C_{x}} d\psi_{x}},~~~~~~~
C_{x}= \int e^{- U(\psi_{x})} d\psi_{x},
\end{equation}
and
\begin{equation}
U_{x} = \varepsilon \gamma_{x}^{-1}\left(\frac{1}{2}q_{x}^{6} + q_{x}^{3}p_{x} + [M + 2\zeta c_{1}]p_{x}^{2} + \lambda^{-1/3}Mq_{x}^{4}\right) .
\end{equation}
To be precise, i.e., to take care of details for the sites at the ends, for $x_{0}= \varepsilon$ and $x_{0} = \mathfrak{T}$ we define
\begin{eqnarray}
d\n(\psi_{x}(x_{0}=\varepsilon)) &=&  \exp\left\{-U_{x}(x_{0}=\varepsilon) -
\varepsilon\gamma_{x}^{-1}[M + 2\zeta c_{1}]p_{\vec{x}}^{2}(x_{0}=\varepsilon)\right\} \times
\nonumber \\
& &\times d\psi_{x}(x_{0}=\varepsilon)dp_{\vec{x}}(x_{0}=\varepsilon)/{\rm normalization} ~,
\nonumber \\
d\n(\psi_{x}(x_{0}=\mathfrak{T})) &=&  \exp\left\{ -\varepsilon\gamma_{x}^{-1}\left[q_{\vec{x}}^{6}(x_{0}=\mathfrak{T}) +
\lambda^{-1/3}Mq_{\vec{x}}^{4}(x_{0}=\mathfrak{T})\right]\right\} \times
\nonumber \\
& & \times dq_{\vec{x}}(x_{0}=\mathfrak{T})/{\rm normalization}~.
\end{eqnarray}

Hence, see Eq.(\ref{principal}), the partition function $Z_{\Lambda}$ is given by
\begin{equation}
Z_\La = C^{\Lambda}\prod_{x\in \La}\left( \int d\nu(\psi_x) \prod_{\{x,y\}\subset \La}
e^{G_{xy}(\psi_x, \psi_y)}\right) ,
\end{equation}
where, $C^{\Lambda} \equiv \prod_{x\in\Lambda}C_{x}$, and $G_{x y} = -\sum_{k=1}^{6}G^{(k)}_{x y}$,
\begin{eqnarray}
& &  G^{(1)}_{x y} = A^{(1)}_{x y}q_x p_y = \varepsilon\g^{-1}_x J_{\vec{x} \vec{y}}(1-\d_{\vec{x} \vec{y}}) \l^{-1/3} \d_{x_0, y_0} q_x p_y ~,
~ G^{(2)}_{x y} = A^{(2)}_{x y}q_x q_y = \varepsilon\g^{-1}_x J_{\vec{x} \vec{y}}(1-\d_{\vec{x} \vec{y}}) \l^{-2/3} M \d_{x_0, y_0} q_x q_y ~, \nonumber \\
& &  G^{(3)}_{x y} = A^{(3)}_{x y}q_x q_y = \varepsilon\sum_{\vec{k}\atop \vec{k}\neq\vec{x},\vec{y}} \frac{\g^{-1}_k  \l^{-2/3}}{4} J_{\vec{x} \vec{k}} J_{\vec{k} \vec{y}} \d_{x_0, y_0} q_x q_y ~,
~ G^{(4)}_{x y} = A^{(4)}_{x y}q_x^3 q_y = \varepsilon\g^{-1}_x J_{\vec{x} \vec{y}}(1-\d_{\vec{x} \vec{y}}) \l^{-1/3} \d_{x_0, y_0} q_x^3 q_y ~, \nonumber \\
 & &  G^{(5)}_{x y} = A^{(5)}_{x y}q_x q_y = 2\varepsilon\g^{-1}_x \l^{-2/3} M \d_{\vec{x} \vec{y}} [\d_{x_0, y_0} - c_1(\D(x_0, y_0))] q_x q_y ~,
~ G^{(6)}_{x y} = A^{(6)}_{x y}p_x p_y = -\varepsilon\g^{-1}_x \zeta c_1 \d_{|x_0-y_0|,\varepsilon} \d_{\vec{x} \vec{y}}p_{x}p_{y} .\nonumber
\end{eqnarray}

We rewrite the partition function  as
\begin{equation}
Z_\La = C^{\La}\prod_{x\in \La} \int d\nu(\psi_x) \prod_{\{x,y\}\subset \La}
\left(e^{G_{xy}(\psi_x, \psi_y)}-1+1\right)= C^{\La}\Xi_\La, \label{cluster1}
\end{equation}
with
\begin{equation}
\Xi_\La= 1 +
\sum_{n\geq 1}{1\over n!}
\sum_{R_1, \dots, R_n\subset\La\atop
R_i\cap R_j=\emptyset,~|R_i|\ge 2}
 \rho(R_1)\dots  \rho(R_n), \label{cluster2}
\end{equation}
where $R_1, \dots R_n\subset\L$ is a collection of subsets
of $\L$ with cardinality greater than 1,
with associated activities $\r(R)$ given by
\be
\label{rhoR}
\rho(R)= \prod_{x\in R}\int d\nu(\psi_x)
\sum_{g\in G_{R}}
\prod_{\{x,y\}\subset \La}(e^{ G_{xy}(\psi_x, \psi_y)}-1),
\ee
where $\sum_{g\in G_{R}}$ is the sum over the connected
graphs on the set $R$. Given a finite set $A$, we define a graph $g$ in A
as a collection $\{\g_1,\dots,\g_m\}$ of distinct pairs of $A$, i.e.,
$\g_i=\{x_i,y_i\}\subset A$ with $x_i\neq y_i$.
A graph $g=\{\g_1,\dots,\g_m\}$ in $A$ is connected
if for any $B$,$C$ of subsets of $A$ such that $B\cup C=A$ and
$B\cap C =\emptyset$, there is a $\g_i\in g$ such that
$\g_i \cap B\ne\emptyset$ and $\g_i \cap C\ne\emptyset$.
The pairs $\g_i$ are called links of the graph. We denote by $|g|$ the
number of links in $g$.

From the standard polymer theory \cite{Si,GJ}, it follows that we can expand $\log\Xi_\L$ as
\be
\label{LogT}
\log\Xi_\L=
\sum_{n\geq 1}{1\over n!}
\sum_{R_1, \dots R_n\subset\L\atop
|R_i|\ge 2}\phi^T(R_1 ,\dots R_n)
\r(R_1)\dots  \r(R_n),
\ee
with
\ben
\label{phiT}
\phi^{T}(R_{1},\dots ,R_{n})=\left\{
\begin{array}{lcl}
1 & , & \mbox{ if } n=1,
\\[0.2cm]
\sum_{f\in G_{n}\atop f\subset g(R_1 ,\dots ,R_n)}(-1)^{|f|} & , & \mbox{ if } n\ge 2 \mbox{ and }
g(R_1 ,\dots ,R_n)\in G_n,
\\[0.3cm]
0 & , & \mbox{ if } g(R_1 ,\dots ,R_n)\notin G_n,
\end{array}\right.
\een
where $G_n$ above denotes the set of the connected graphs on $\{1,\dots ,n\}$,
and $g(R_1 ,\dots ,R_n)$ denotes the
graph in $\{1,2,\dots ,n\}$ which has the link $\{i,j \}$ if
and only if
$R_i\cap Rj\neq \emptyset$.

The connection between the polymer expansion and a perturbative series is clear from the expressions (\ref{cluster1}, \ref{cluster2}, \ref{LogT}) above. To make it explicit and transparent, note that if we
write $G_{xy}$ as $\beta G_{xy}$, then $\rho = \mathcal{O}(\beta)$ (it involves $\exp[\beta G_{xy}] -1$), and the polymer expansions above give us a power series in $\beta$.

A sufficient condition for the convergence of the polymer series (\ref{LogT}), uniformly in $\L$, is given by the well known result due to Koteck\'y and Preiss:
\begin{lemma} \label{lemma:KP}
If there is $a>0$ such that
\begin{equation}
\sup_{x\in \L}\sum_{x \in R \subset \Lambda} |\r(R)|e^{a|R|}<a, \label{CondConv}
\end{equation}
then
\begin{equation}
| \log\Xi_\Lambda | \leq \sum_{n\geq 1}{1\over n!}
\sum_{R_1, \dots R_n\subset\L\atop
|R_i|\ge 2} \left|\phi^T(R_1 ,\dots R_n)
\r(R_1)\dots  \r(R_n)\right| \leq E_{a}|\Lambda| ,
\end{equation}
where $E_{a}$ does not depend on $\Lambda$.
\end{lemma}

\begin{proof}
See Refs.\cite{KP,B,C}.
\end{proof}

In relation to the two-point correlation (\ref{principal}), we can rewrite it as
\be
\label{S2x1x2}
S_2(x_1;x_2)=
{\partial^2\over
\partial\a_1\partial\a_2}
\log \left.\tilde \Xi_\L(\a_1 ,\a_2)\right|_{\a=0},
\ee
with
\be
\tilde \Xi_\L(\a_1 ,\a_2)=\prod_{x\in \L}\int d\n(\psi_x)
e^{ G_{xy}(\psi_x, \psi_y)}
(1+\a_1 \psi_{x_1}^{(c)})(1+\a_2 \psi_{x_2}^{(c)}),
\ee
where $\psi_{x}^{(c)} = q_{\vec{x}}(x_{0})$ or $p_{\vec{x}}(x_{0})$.
Note that $\tilde \Xi_\L(\a_1=0 ,\a_2=0)=\Xi_\L$. Again, we
expand $\tilde \Xi_\L(\a_1,\a_2)$ in terms of polymers.
For any $R\subset\L$, we denote by $I_R$ the subset
(possibly empty) of $\{1,2\}$ such that $i\in I_R$ iff
$x_i\in R$, where $i=1,2$.
We have
\be
\label{XiL}
\Xi_\L(\a_1,\a_2)=1+\sum_{n\geq 1}{1\over n!}
\sum_{{R_1, \dots R_n\subset\L\atop
R_i\cap R_j=\emptyset~|R_i|\ge 1}}
\tilde \r(R_1,\a)\dots \tilde \r(R_n,\a),
\ee
where
\begin{equation}
\tilde \r(R,{ \a})= \left\{
\begin{array}{lcl}
\prod_{x\in R}\int d\n(\psi_x)
\prod_{i\in I_R}(1+\a_i\psi_{x_i}^{(c)})
\sum\limits_{g\in G_{R}}
\prod\limits_{\{ x,y\}\in g}(e^{ G_{xy}(\psi_x,\psi_y)}-1)
& , & \mbox{ for } |R |\geq 2,
\\[0.2cm]
\prod_{x\in R}\int d\n(\psi_x)
\prod\limits_{i\in I_R}\a_i \psi_{x_i}^{(c)} & , & \mbox{ for } I_R\neq \emptyset, |{R}|=1,
\\[0.2cm]
0 & , & \mbox{ for } I_R=\emptyset, |R|=1.
\end{array}
\right.
\end{equation}


It is important to note that the one-body polymers $R=\{x\}$ can also contribute to the
partition function (\ref{XiL}), but only if $x=x_i$ for some
$i\in \{1,2\}$.

If we take the $\log$ of (\ref{XiL}) and
note that only the terms proportional to
$\a_1\a_2$ give a non vanishing
contribution to the two-point truncated correlation function, we obtain
\be
\label{SerieS2}
S_2(x_1;x_2)=\sum_{n\geq 1}{1\over n!}\sum_{i_1,i_2=1}^{n}
\sum_{R_1, \dots R_n\subset\L,~|R_j|\ge 2\atop R_{i_1}\ni x_1
R_{i_2}\ni x_2}
\phi^T(R_1 ,\dots R_n)
\tilde \r(R_1)\dots \tilde \r(R_n),
\ee
where
\begin{eqnarray}
\label{rhotil}
\tilde \r(R_i)&=&
\prod_{x\in R_i}\int d\n(\psi_x)
\left[\left(\psi_{x_1}^{(c)}\right)^{\b_i^{1}}+\b_i^{1}\ck_{1}\right]\left[\left(\psi_{x_2}^{(c)}\right)^{\b_i^{2}}+\b_i^{2}\ck_{2}\right]
\times \\
& & \times
\sum\limits_{g\in G_{R}}
\prod\limits_{\{ x,y\}\in g}(e^{  G_{xy}(\psi_x, \psi_y)}-1),
\nonumber
\end{eqnarray}
with $\b_i^j = 0$ if $i\neq i_j$, or $1$ if $i= i_j$ and
$\ck_{k}=\int d\n(\psi)\psi^{(c)}_{x_{k}}$.
Due to the fact that $R_1,\dots ,R_n$ must
be connected, the
one-body polymers are absorbed in the activity
of the many body polymers in the terms proportional to $\ck$ above.
 And so, each 1-body polymer (if any)
is always contained in, at least, one many-body polymer.


\section{Convergence of the polymer expansion} \label{sec:conv}

We will describe, in detail, the case in which the interparticle potential
$J_{\vec x,\vec y}$ has an integrable
 polynomial decay (cases with exponential decay or with finite range can be treated in a similar, but easier,  way).
Precisely, we assume here that, for $\vec{x}\neq\vec{y}$,
\begin{equation}
{J'\over |\vec x - \vec y|^{p}} \leq J_{\vec x,\vec y} \leq  {J\over |\vec x - \vec y|^{p}},
\end{equation}
where $J, J'$ are real constants and  $p\ge 1+ \epsilon$ with $\epsilon >0$.

In what follows, we assume the regime of large dissipation, that means $\zeta$ large (and so,
large $\alpha= \zeta/2$ and large harmonic pinning constant $M= 3\alpha^2$, see appendix \ref{apendice}),
and, more importantly,  we also assume the regime of large anharmonicity, i.e., at the end we take $\lambda$ as large
as necessary.

Our strategy is to prove that the Koteck\'y-Preiss's condition (\ref{CondConv}) is satisfied by our specific polymer expansion.
But, before carrying out any computation, in order to control
 the exaggerated number of graphs that appears in the expression for the activity  $\r(R)$, we recall an important and well known result
 (to be used ahead), namely, the  Brydges-Battle-Federbush tree graph inequality:

\begin{lemma}
\label{3graph}
Let $R$ be a finite set with cardinality $|R|$ and
let $\{V_{xy}:\{x,y\} \subset R\}$ be a set with $|R|(|R|-1)/2$ real numbers
(precisely, $\{x,y\}$ are unordered pairs in $R$).
Suppose that there exist $|R|$ positive numbers
$V_{x}$ (with $x\in R$)
such that, for any subset $S\subset R$,
\be
\left[ \sum_{x\in S}V_{x}+\sum_{\{x,y\}\in S}V_{xy}\right] \geq 0.
\ee
Then
\be
\left|\sum_{g\in G_R}\prod_{\{x,y\}\in g}
\left(e^{-V_{xy}}-1\right)\right|\leq
e^{\sum_{x\in R}V_{x}}
\sum_{\t\in T_R}\prod_{\{x,y\}\in \t} |V_{xy}|,
\ee
where $T_R$ denotes the set of the tree graphs on $R$.
\end{lemma}

\begin{proof}
See Refs.\cite{B,BF}.
\end{proof}

To ensure the Koteck\'y-Preiss's condition (\ref{CondConv}),  we will bound the factor
\be
\e_n(z,z') = \sum_{R\subset \La: |R|=n\atop z,z' \in R}|\r(R)|,
\ee
and show that $\e_n(z,z')\le [f(J, \l^{-1}, c_1)]^n$,
where $f(J, \l c_1)\to 0$ as
$\l^{-1}, J, c_1,\to 0$.

Note that for $\varphi, \psi \in \R$, $x \in \La$, we have
\ben
&& 2|\varphi||\psi| \leq \varphi^2 + \psi^2 ~,~~
 |\varphi|^{a}|\psi|^b \leq |\varphi|^{a+b} + |\psi|^{a+b} ~~ {\rm for}~ a,b \geq 0,\\
&&  \sum_{y \in \La} \d_{|x_0-y_0|,\varepsilon} \d_{\vec{x} \vec{y}} \leq \sum_{y_0 \in \Z} \d_{|x_0-y_0|,\varepsilon} \leq 2 ~,~~
 \sum_{y \in \La} J_{\vec{x} \vec{y}} \d_{x_0, y_0} \leq \sup_{\vec{x} \in \Z} \sum_{\vec{y} \in \Z}|J_{\vec{x} \vec{y}}| \leq J_{M} ,\\
&&  \sum_{k \in \Z\atop \vec{k}\neq\vec{x},\vec{y}} J_{\vec{x} \vec{k}}J_{\vec{k} \vec{y}} \leq \sum_{k \in \Z\atop \vec{k}\neq\vec{x},\vec{y}} \frac{J}{|\vec{x}-\vec{k}|^p}\frac{J}{|\vec{k}-\vec{y}|^p} \leq \frac{J^2 \mathcal{O}(1)}{|\vec{x}-\vec{y}|^p}(1-\d_{\vec{x},\vec{y}}) .
\een
By using the bound $\g^{-1}_x \leq \g^{-1} = (2 \z  T_{min})^{-1}$, where $T_{min} = \min_{x}\{T_x\}$, we obtain
\be
 \left|\sum_{\{x,y\} \subset R}G^{(k)}_{x y}\right| &\leq& \sum_{\{x,y\} \subset R}|G^{(k)}_{x y}| \nonumber \\
 &\leq& \sum_{\{x,y\} \subset R}|A^{(k)}_{x y}|\varphi_x|^{a}|\psi_y|^{b}
\leq \sum_{\{x,y\} \subset R}|A^{(k)}_{x y}| (|\varphi_x|^{a+b}+|\psi_y|^{a+b})
\nonumber \\
 &\leq& \sum_{x \in R}(|\varphi_x|^{a+b}+|\psi_x|^{a+b}|) \sum_{y \in R}|A^{(k)}_{x y}|.
\ee
Hence,
\ben
&& \left|\sum_{{x,y} \in \La}G^{(1)}_{x y}\right| \leq \sum_{x \in R} \varepsilon \g^{-1}J_M\l^{-1/3}\frac{q_x^2+p_x^2}{2} ,
~~ \left|\sum_{{x,y} \in \La}G^{(2)}_{x y}\right| \leq \sum_{x \in R} \varepsilon M\g^{-1}J_M\l^{-2/3}q_x^4 ,
\\
&& \left|\sum_{{x,y} \in \La}G^{(3)}_{x y}\right| \leq \sum_{x \in R} \varepsilon \g^{-1}J_M^2\l^{-2/3}\mathcal{O}(1)q_x^2 ,
~~ \left|\sum_{{x,y} \in \La}G^{(4)}_{x y}\right| \leq \sum_{x \in R} 2\varepsilon \g^{-1}J_M^2\l^{-1/3}q_x^4 ,
\\
&& \left|\sum_{{x,y} \in \La}G^{(5)}_{x y}\right| \leq \sum_{x \in R} \varepsilon M\g^{-1}J_M^2\l^{-2/3}(1+4c_1)\frac{q_x^2}{2} ,
~~ \left|\sum_{{x,y} \in \La}G^{(6)}_{x y}\right| \leq \sum_{x \in R} 2 \varepsilon \g^{-1}_x \zeta c_1p_x^2.
\een
And so,
\be
\left|\sum_{\{x,y\} \subset R}G_{x y}\right| \leq \sum_{x \in R} \mathcal{P}(q_x,p_x) ,
\ee
where $\mathcal{P}(q_x,p_x)$ is a polynomial of degree 4 in $q_x$ and 2 in $p_x$, and it is bounded from below. Hence, there are constants $C_1, C_2$ and $C_3$ depending on $\varepsilon, \l, J, M, \g$ such that $\mathcal{P}(q_x,p_x) \leq C_1 q_x^4 + (C_2+2\varepsilon\g^{-1}_x \zeta c_1) p_x^2 + C_3$. By using the Brydges-Battle-Federbush tree graph inequality, we get
\be
\left|\sum_{g\in G_{R}}
\prod_{\{x,y\}\in g}(e^{G_{xy}}-1)\right|
\le
\prod_{x\in R}e^{\mathcal{P}(q_x,p_x)}
\sum_{\t\in T_{R}}
\prod_{\{x,y\}\in \t}|G_{xy}|.
\ee
Consequently,
\ben
 \lefteqn{\sum_{R\subset \La: |R|\ge 2\atop z,z' \in R}|\rho(R)| e^{|R|}
 = \sum_{n\ge 2} e^n \sum_{R\subset \La:|R|=n\atop z,z' \in R}|\rho(R)|} \\
 &=& \sum_{n\ge 2} \frac{e^n}{(n-2)!}
 \sum_{x_1,\dots,x_n \in \La:\atop x_1=z,x_2=z',x_i\ne x_j}
 \!\!\!\!\!\!\! |\rho(R=\{x_1,\dots,x_n\})|
\\
 &\le& \sum_{n\ge 2} \frac{e^n}{(n-2)!}
 \sum_{x_1,\dots,x_n \in \La:\atop x_1=z,x_2=z',x_i\ne x_j}
  \int \prod_{i=1}^{n} d\nu(\psi_{x_i}) e^{\mathcal{P}(q_x,p_x)}
 \sum_{\tau\in T_n}
 \prod_{\{i,j\}\in \tau}|G_{{x_i}{y_j}}|.\\
\een
Recall now that $|\tau|=n-1$. Hence, fixing $\tau \in T_n$, we have
\be
\prod_{\{i,j\}\in \tau}| G_{{x_i}{y_j}}|
&\le& \prod_{\{i,j\}\in \tau} \sum_{s=1}^{6} |A_{x_i x_j}^{(s)}|
|q_{x_i}^{a^{(i)}_{s}}||p_{x_i}^{b^{(i)}_{s}}||q_{x_i}^{a^{(j)}_{s}}||p_{x_j}^{b^{(j)}_{s}}| \nonumber\\
&\le& \sum_{\{i,j\}\in \tau} \sum_{s_{ij}=1}^{6}
\prod_{k=1}^{n}|q_{x_k}|^{n_k(s)}|p_{x_k}|^{m_k(s)}
\prod_{\{i,j\}\in \tau}|A_{x_i x_j}^{(s_{ij})}|,
\ee
where $\{s_{ij}\}_{\{i,j\}\in \tau}$ is a sequence of possible choices of $s$'s (from 1 to 6), for each line $\{i,j\} \in \tau$. The exponents $n_k(s)$ and $m_k(s)$ depend on such sequence, and depend also on the exponents $a_s$ and $b_s$ of $q_{x_i}$ and $p_{x_i}$ in each $G^{(s)}_{x y}$. In any case, we have the bounds $0 \leq n_k(s) \leq d_k \cdot \max\{a_s\}  \leq 3d_k$ and $0 \leq m_k(s) \leq d_k \cdot \max\{b_s\} \leq d_k$, where $\{d_k\}_{k=1}^{n}$ are the incidence indices of the tree $\tau \in T_n$, with $1\le d_k\le n-1$ and $\sum_{k=1}^{n}d_k=2n-2$. Then
\be
&& \sum_{R\subset \La: |R|\ge2\atop z,z' \in R}|\rho(R)| e^{|R|} \leq
   \sum_{n \ge 2} \frac{e^n}{(n-2)!}
 \sum_{x_1,\dots,x_n \in \La:\atop x_1=z,x_2=z',x_i\ne x_j}
 \int \prod_{i=1}^{n} d\nu(\psi_{x_i}) e^{\mathcal{P}(q_{x_i},p_{x_i})} \times
\nonumber\\
&& \times  \sum_{\tau\in T_n}
 \sum_{\{i,j\}\in \tau} \sum_{s_{ij}=1}^{6}
 \prod_{k=1}^{n}|q_{x_k}|^{n_k(s)}|p_{x_k}|^{m_k(s)}
 \prod_{\{i,j\}\in \tau}|A_{x_i x_j}^{(s_{ij})}| \leq
\nonumber\\
&& \le \sum_{n \ge 2} \frac{e^n}{(n-2)!}
 \sum_{x_1,\dots,x_n \in \La:\atop x_1=z,x_2=z',x_i\ne x_j}
  \sum_{\tau\in T_n}
 \sum_{\{i,j\}\in \tau} \sum_{s_{ij}=1}^{6}
\nonumber\\
&& \prod_{k=1}^{n} \left(
 \int d\nu(\psi_{x_k}) e^{\mathcal{P}(q_{x_i},p_{x_i})}
  |q_{x_k}|^{n_k(s)}|p_{x_k}|^{m_k(s)}\right)
 \prod_{\{i,j\}\in \tau}|A_{x_i x_j}^{(s_{ij})}|.
\ee

\begin{lemma}
\label{Cota da integral}
$\forall \alpha, \beta > 0$, $\alpha, \beta \in \R$, $C_1<\frac{\varepsilon\g^{-1}}{6}$ and $C_2<\frac{\varepsilon\g^{-1}}{4}$, we have
\ben
\int d\nu(\psi)|q|^{\a}|p|^{\b}e^{\mathcal{P}(q,p)} \leq \frac{\g^{-1/2} e^{K_5-K_2}K_1^{1/6}(1+c_1)^{1/2}}{3K_3^{\frac{1+\a}{6}}K_4^{\frac{1+\b}{2}}}
\Gamma\left(\frac{1+\a}{6}\right)\Gamma\left(\frac{1+\b}{2}\right),
\een
where $K_1,..., K_5$, are constants such that, $\forall q, p \in \R$,
\ben
\g^{-1}\left\{\left[\frac{1}{2}-\frac{1}{4(M +2 \zeta c_1)}\right]q^6+\l^{-1/3}Mq^4\right\} \leq K_1q^6+K_2,
\een
and
\ben
U(q, p)-\mathcal{P}(q,p) \geq K_3q^6+K_4p^2+K_5 .
\een
\end{lemma}
\begin{proof}
We have
\be
U(q,p) &=& \varepsilon\gamma^{-1}\left(\frac{1}{2}q^{6} + q^{3}p + (M +2 \zeta c_1)p^{2} + \lambda^{-1/3}Mq^{4}\right)
\nonumber\\
&=& \varepsilon\g^{-1}\left\{\left[\frac{q^3}{2(M +2 \zeta c_1)^{1/2}}+(M +2 \zeta c_1)^{1/2}p\right]^2
-\frac{q^6}{4(M +2 \zeta c_1)}+  \frac{q^6}{2}+\lambda^{-1/3}Mq^{4}\right\}
\nonumber\\
&\leq& \varepsilon\g^{-1}\left[\frac{q^3}{2(M +2 \zeta c_1)^{1/2}}+(M +2 \zeta c_1)^{1/2}p\right]^2
+K_1q^6+K_2.
\ee
It also follows that
\be
U(q,p) &=& \varepsilon\gamma^{-1}\left(\frac{1}{2}q^{6} + q^{3}p + (M +2 \zeta c_1)p^{2} + \lambda^{-1/3}Mq^{4}\right) \nonumber\\
 &\geq &  \varepsilon\gamma^{-1}\left(\frac{1}{2}q^{6} + q^{3}p + (1 +2 \zeta c_1)p^{2} + \lambda^{-1/3}Mq^{4}\right)\nonumber\\
&=& \varepsilon\g^{-1}\left[\frac{q^6}{6}+\left(\frac{1}{\sqrt3}q^3+\frac{\sqrt3}{2}p\right)^2+\frac{p^2}{4}+2\zeta c_1p^2\right]
 \geq  \varepsilon\g^{-1}\left[\frac{q^6}{6}+\frac{p^2}{4}+2\zeta c_1p^2\right] .
\ee
Hence,
\be
U(q, p)-\mathcal{P}(q,p)
&\geq& \left(\frac{\varepsilon\g^{-1}}{6}-C_1\right)q^6+\varepsilon\lambda^{-1/3}Mq^{4}+\left(\frac{\varepsilon\g^{-1}}{4}-C_2\right)p^2-C_3 \nonumber\\
&\geq& K_3q^6+K^4p^2+K_5,
\ee
with $K_3, K_4 > 0$, as $C_1<\frac{\varepsilon\g^{-1}}{6}$ and $C_2<\frac{\varepsilon\g^{-1}}{4}$. From the definition of the s.s.d., we have
\ben
\int d\nu(\psi)|q|^{\a}|p|^{\b}e^{\mathcal{P}(q,p)} = \frac{1}{C_I} \int d\psi|q|^{\a}|p|^{\b}e^{\mathcal{P}(q,p)-U(q,p)}.
\een
And so,
\ben
C_I &=& \int_{-\infty}^{\infty}\int_{-\infty}^{\infty} e^{-U(q,p)}dp dq \\
&& \geq \int_{-\infty}^{\infty}\int_{-\infty}^{\infty} \exp\left\{-\varepsilon\g^{-1}\left[\frac{q^3}{2(M +2 \zeta c_1)^{1/2}}+(M +2 \zeta c_1)^{1/2}p\right]^2 \right.  - K_1q^6-K_2 \bigg\} dp\ dq
\\
&=& \int_{-\infty}^{\infty} \sqrt{\frac{\pi \g}{\varepsilon(M +2 \zeta c_1)}}e^{-K_1q^6-K_2} dq = 2\sqrt{\frac{\pi \g}{\varepsilon(M +2 \zeta c_1)}}e^{-K_2}K_1^{-1/6}
\\
&\geq& \sqrt{\frac{\g}{\varepsilon(M +2 \zeta c_1)}}e^{-K_2}K_1^{-1/6}\Gamma\left(\frac{7}{6}\right),
\een
where the last inequality comes from $2\sqrt{\pi}\Gamma\left(\frac{7}{6}\right)\approx 3.3 > 1$. We still have
\ben
\int d\psi|q|^{\a}|p|^{\b}e^{\mathcal{P}(q,p)-U(q,p)} &\leq& \int_{-\infty}^{\infty}\int_{-\infty}^{\infty} e^{K_3q^6+K_4p^2+K_5} dp dq =
\\
&=&
\frac{1}{3} e^{-K_5} K_3^{-\frac{1+\alpha}{6}} K_4^{-\frac{1+\beta}{2}} \Gamma\left(\frac{\alpha+1}{6}\right) \Gamma\left(\frac{\beta +1}{2}\right).
\een
And the lemma's proof follows from these two bounds.
\end{proof}

 Using the fact that for large $x,y$, e.g.  $x,y >1$, there exists a constant $c$, such that $\Gamma(x)\Gamma(y)\leq c \Gamma(x+y-1)$, it follows that
\ben
\Gamma\left(\frac{1+n(k)}{6}\right)\Gamma\left(\frac{1+m(k)}{2}\right) \leq K_6\Gamma \left(\frac{1+n(k)}{6}+\frac{1+m(k)}{2}-1\right) \leq K_6\Gamma(d_k),
\een
for a positive constant $K_6$, chosen to take care of possible small values of $x,y$ in $\Gamma(x)\Gamma(y)$. Let $\tilde{K}_3=\min \{1,K_3^{1/2}\}$ and $\tilde{K}_4=\min \{1,K_4^{1/2}\}$. So, $K_3^{\frac{n(k)}{6}} \geq \tilde{K}_3^{d_k}$ and $K_4^{\frac{m(k)}{2}} \geq \tilde{K}_3^{d_k}$. Using this with the lemma, we obtain
\ben
\sum_{R\subset \La: |R|\ge 2\atop z,z' \in R}|\rho(R)| e^{|R|}
&\le& \sum_{n \ge 2} \frac{e^n}{(n-2)!}
 \sum_{x_1,\dots,x_n \in \La:\atop x_1=z,x_2=z',x_i\ne x_j}
 \sum_{\tau\in T_n}
 \sum_{\{i,j\}\in \tau} \sum_{s_{ij}=1}^{6}
\\
 & & \times
 \prod_{k=1}^{n} \left[\frac{\varepsilon^{1/2}\g^{-1/2} e^{K_5-K_2}K_1^{1/6}(M +2 \zeta c_1)^{1/2}}{3K_3^{\frac{1}{6}}K_4^{\frac{1}{2}}} \tilde{K}_3^{-d_k}\tilde{K}_4^{-d_k}K_6\Gamma(d_k)\right]
\\
& & \times
 \prod_{\{i,j\}\in \tau}|A_{x_i x_j}^{(s_{ij})}|
\een
\be
&\le& \sum_{n \ge 2} \frac{e^n}{(n-2)!}
 \left(\frac{\varepsilon^{1/2}\g^{-1/2} e^{K_5-K_2}K_1^{1/6}(M +2 \zeta c_1)^{1/2}K_6}{3K_3^{\frac{1}{6}}K_4^{\frac{1}{2}}}\right)^n
 \tilde{K}_3^{-2n+2}\tilde{K}_4^{-2n+2}
\nonumber\\
\label{Soma1}
 & & \times
 \sum_{\tau\in T_n}
 \sum_{\{i,j\}\in \tau} \sum_{s_{ij}=1}^{6}
\prod_{k=1}^{n} \Gamma(d_k)
 \sum_{x_1,\dots,x_n \in \La:\atop x_1=z,x_2=z',x_i\ne x_j}
 \prod_{\{i,j\}\in \tau}|A_{x_i x_j}^{(s_{ij})}|,
\ee
where we used $\prod_{k=1}^n\phi^{d_k}=\phi^{d_1+d_2+\dots+d_n}=\phi^{2n-2}$.

We note that for any $\tau \in T_n$, there is a unique path
$\bar{\tau}$ in $\tau$ which joins vertex 1 to vertex 2.
Fixing $\tau \in T_n$, let be $\bar{\tau} \equiv
\{1,i_1\},\{i_1,i_2\},\{i_2,i_3\},\dots,$
\linebreak
$\{i_{k-1},i_k\},\{i_k,2\}$
and $I_{\tau} \equiv \{1,i_1,i_2,\dots,i_k,2\}$ the subset of
$\{1,2,3,\dots,n\}$ whose elements are the vertices of the path
$\bar{\tau}$. Hence, $|\tau|=n-1$, $|\bar{\tau}|=k+1$ and
$|\tau\setminus\bar{\tau}| = n-k-2$.\\

From the definitions of $|A_{xy}^{(s)}|:s=1,\dots, 6$,
we see that all the terms vanish if $\md{x_0-y_0}>\varepsilon$.
Hence, fixing $\tau \in T_n$, if  $\exists ~\{i,j\} \in \tau$ such that
$\md{(x_i)_0-(x_j)_0}>\varepsilon$ we have
$|A_{x_i x_j}^{(s)}|=0 ~\forall s=1,\dots,(c+1)$, and so, this tree
$\tau$ does not contribute to the sum (\ref{Soma1}).
Then, given $\md{(x_1)_0-(x_2)_0}$, as $|\bar{\tau}|=k+1$, if
$\md{(x_1)_0-(x_2)_0}>\varepsilon(k+1)$ then $\exists ~\{i,j\} \in \bar{\tau}$
such that $\md{(x_i)_0-(x_j)_0}>\varepsilon$, and so, $\tau$ does not contribute to
(\ref{Soma1}).
As $\bar{\tau}\subset\tau$ we have $n-1\ge k+1$. Therefore,
any tree $\tau\in T_n$ such that
$\md{(x_1)_0-(x_2)_0}>\varepsilon(n-1)\ge \varepsilon(k+1)$ does not contribute to (\ref{Soma1}),
in other words, $\rho(R)$ vanishes if $\md{(x_1)_0-(x_2)_0}>\varepsilon(|R|-1)$,
i.e., if $|R| < \frac{\md{(x_1)_0-(x_2)_0}}{\varepsilon}+1$.\\

We define $\mathcal{N}'\equiv {\rm max}\left\{\frac{\md{z_0-{z'}_0}}{\varepsilon}+1, 2\right\}$, and we have
\be
\label{contaN}
\sum_{R\subset \La: |R|\ge 2\atop z,z' \in R}|\rho(R)| e^{|R|} =
\sum_{R\subset \La: |R|\ge \mathcal{N}'\atop z,z' \in R}|\rho(R)| e^{|R|}.
\ee

Now, we note that
\ben
\d_{\vec{x},\vec{y}} \d_{|x_0-y_0|,\varepsilon} \le
\d_{\vec{x},\vec{y}} ~e^{-\frac{|x_0-y_0|}{\varepsilon}+1},
~~~~~\d_{x_0,y_0}\le e^{-\frac{|x_0-y_0|}{\varepsilon}+1}.
\een

Then,
\ben
|A^{(1)}_{x y}| &=& \varepsilon\g^{-1} J_{\vec{x} \vec{y}}(1-\d_{\vec{x} \vec{y}}) \l^{-1/3} \d_{x_0, y_0} \leq
\\
&\leq&
\varepsilon\g^{-1} J \l^{-1/3} e^{-\frac{|x_0-y_0|}{\varepsilon}+1}\left(\frac{1-\d_{\vec{x} \vec{y}}}{|\vec{x}-\vec{y}|^p}\right)=
A_1 e^{-\frac{|x_0-y_0|}{\varepsilon}+1}\left(\frac{1-\d_{\vec{x} \vec{y}}}{|\vec{x}-\vec{y}|^p}\right),
\\
|A^{(2)}_{x y}| &=& \varepsilon\g^{-1} J_{\vec{x} \vec{y}}(1-\d_{\vec{x} \vec{y}}) \l^{-2/3} M \d_{x_0, y_0} \leq
\\
&\leq&
\varepsilon\g^{-1} J \l^{-2/3} M e^{-\frac{|x_0-y_0|}{\varepsilon}+1}\left(\frac{1-\d_{\vec{x} \vec{y}}}{|\vec{x}-\vec{y}|^p}\right)=
A_2 e^{-\frac{|x_0-y_0|}{\varepsilon}+1}\left(\frac{1-\d_{\vec{x} \vec{y}}}{|\vec{x}-\vec{y}|^p}\right),
\\
|A^{(3)}_{x y}| &=& \sum_{\vec{k}\atop \vec{k}\neq\vec{x},\vec{y}} \frac{\varepsilon\g^{-1}  \l^{-2/3}}{4} J_{\vec{x} \vec{k}} J_{\vec{k} \vec{y}}(1-\d_{\vec{x} \vec{y}}) \d_{x_0, y_0} \leq
\\
&\leq&
\frac{\varepsilon\g^{-1}  \l^{-2/3}}{4} J J_M \mathcal{O}(1) e^{-\frac{|x_0-y_0|}{\varepsilon}+1}\left(\frac{1-\d_{\vec{x} \vec{y}}}{|\vec{x}-\vec{y}|^p}\right)
=
A_3 e^{-\frac{|x_0-y_0|}{\varepsilon}+1}\left(\frac{1-\d_{\vec{x} \vec{y}}}{|\vec{x}-\vec{y}|^p}\right),
\een
\ben
|A^{(4)}_{x y}| &=& \varepsilon\g^{-1} J_{\vec{x} \vec{y}}(1-\d_{\vec{x} \vec{y}}) \l^{-1/3} \d_{x_0, y_0} \leq
\\
&\leq&
\varepsilon\g^{-1} J \l^{-1/3} e^{-\frac{|x_0-y_0|}{\varepsilon}+1}\left(\frac{1-\d_{\vec{x} \vec{y}}}{|\vec{x}-\vec{y}|^p}\right) =
A_4 e^{-\frac{|x_0-y_0|}{\varepsilon}+1}\left(\frac{1-\d_{\vec{x} \vec{y}}}{|\vec{x}-\vec{y}|^p}\right),
\\
|A^{(5)}_{x y}| &=& 2\varepsilon\g^{-1} \l^{-2/3} M \d_{\vec{x} \vec{y}} |\d_{x_0, y_0} - c_1\D(x_0, y_0)| \leq
\\
&\leq&
2\varepsilon\g^{-1} \l^{-2/3} M (1+3c_1) e^{-\frac{|x_0-y_0|}{\varepsilon}+1}\d_{\vec{x} \vec{y}} =
A_5 e^{-\frac{|x_0-y_0|}{\varepsilon}+1}\d_{\vec{x} \vec{y}},
\\
|A^{(6)}_{x y}| &=& \varepsilon\g^{-1} \zeta c_1 \d_{|x_0-y_0|,1} \d_{\vec{x} \vec{y}} \leq
\varepsilon\g^{-1} \zeta c_1 e^{-\frac{|x_0-y_0|}{\varepsilon}+1}\d_{\vec{x} \vec{y}} =
A_6 e^{-\frac{|x_0-y_0|}{\varepsilon}+1}\d_{\vec{x} \vec{y}}.\\
\een
Hence, for $s=1,\dots,6$,
\ben
|A^{(s)}_{x y}| \leq A_s e^{-\frac{|x_0-y_0|}{\varepsilon}+1}\left(\frac{1-\d_{\vec{x} \vec{y}}}{|\vec{x}-\vec{y}|^p}+\d_{\vec{x} \vec{y}}\right)
 \leq e K F_{xy}^{(1)},
\sup_{x \in \La}\sum_{y \in \La} |A_{xy}^{(s)}|
\leq e\mathcal{O}(1) K,
\een
where we used the notation
\be
K \equiv \max\{A_1, A_2,\dots A_6\},
\ee
and, for $w\in\R$, $w>0$,
\ben
F_{xy}^{(w)} \equiv
 e^{-w\frac{|x_0-y_0|}{\varepsilon}}\left[\frac{(1-\delta_{\vec{x},\vec{y}})}{|\vec{x}-\vec{y}|^p}
 + \delta_{\vec{x},\vec{y}}
  \right].
\een
Then, fixing $\tau \in T_n$ and the sequence $\{s_{ij}\}$, we get
\ben
& &\sum_{x_1,\dots,x_n \in \La:\atop x_1=z,x_2=z',x_i\ne x_j}
 \prod_{\{i,j\}\in \tau}|A_{x_i x_j}^{(s_{ij})}|
= \sum_{x_1,\dots,x_n \in \La:\atop x_1=z,x_2=z',x_i\ne x_j}
 \prod_{\{i,j\}\in \tau\setminus\bar{\tau}}|A_{x_i x_j}^{(s_{ij})}|
 \prod_{\{i,j\}\in \bar{\tau}}|A_{x_i x_j}^{(s_{ij})}|\\
&\le& [e\mathcal{O}(1) K]^{(n-k-2)}
 \sum_{x_{i_1},\dots,x_{i_k} \in \La:
  \atop x_{i_r} \ne x_{i_q} \forall r,q=1,\dots,k}
 \prod_{\{i,j\}\in \bar{\tau}}|A_{x_i x_j}^{(s_{ij})}|\\
&\le&[e\mathcal{O}(1) K]^{(n-k-2)} \!\!\!\!\!\!\!
 \sum_{x_{i_1},\dots,x_{i_k} \in \La:
  \atop x_{i_r} \ne x_{i_q} \forall r,q=1,\dots,k}
 eKF_{x_1 x_{i_1}}^{(1)} eKF_{x_{i_1} x_{i_2}}^{(1)}\dots
  eK F_{x_{i_{k-1}} x_{i_k}}^{(1)} eKF_{x_{i_k} x_2}^{(1)}.\\
\een
Applying iteratively the inequality (for $w_1 < w_2$)
\be
\label{SomaFxy}
\sum_{x_i \in \La:\atop x_i\ne x,y} F_{x x_i}^{(w_1)}F_{x_i y}^{(w_2)}
\le \mathcal{O}(1) F_{xy}^{(w_1)},
\ee
which follows from
\ben
\sum_{\vec{x}_i \in \Z^d:\atop \vec{x}_i\ne \vec{x},\vec{y}}
 \frac{1}{|\vec{x}-\vec{x}_i|^p} \frac{1}{|\vec{x}_i-\vec{y}|^p} \le
 \frac{\mathcal{O}(1)}{|\vec{x}-\vec{y}|^p}, \quad \mbox{and}
\quad \sum_{z_0\in\R :\atop z_0\ne x_0,y_0}
 e^{-w_1\frac{|x_0-z_0|}{\varepsilon}} e^{-w_2\frac{|z_0-y_0|}{\varepsilon}} \le
 \mathcal{O}(1) e^{-w_1\frac{|x_0-y_0|}{\varepsilon}},
\een (the formula is valid for any $w_1<w_2$, in specific for
$w_1=2/3$ and $w_2=1$, which we will take here), we get \be
\label{caso1}
 \sum_{x_1,\dots,x_n \in \La:\atop x_1=z,x_2=z',x_i\ne x_j}
 \prod_{\{i,j\}\in \tau}|A_{x_i x_j}^{(s_{ij})}|
&\le&[e\mathcal{O}(1) K]^{(n-1)}F_{zz'}^{(2/3)}.
\ee

Recall that
\be
\label{SomaTau}
 \sum_{\tau \in T_n}1 = \sum_{d_1+\dots+d_n=2n-2\atop d_i \ge 1}
 \sum_{\tau \in T_n : \atop \tau \approx (d1,\dots,d_n) }1 ~,
\ee
where the notation $\tau \approx (d_1,\dots,d_n)$ means that
the last sum above runs over the trees $\tau \in T_n$ that have
fixed incidence indices $(d_1,\dots,d_n)$. From the Cayley formula
\be
 \sum_{\tau \in T_n : \atop \approx (d1,\dots,d_n) }1 =
 \frac{(n-2)!}{\prod_{i=1}^{n}(d_i-1)!} ~,
\ee
and, fixing $\tau \in T_n$, we have
\ben
 \sum_{\{i,j\}\in \tau} \sum_{s_{ij}=1}^{6} 1 = 6^{n-1}.
\een

Hence, using (\ref{contaN}), we get
\ben
& &\sum_{R\subset \La: |R|\ge 2\atop z,z' \in R}|\rho(R)| e^{|R|}
\\
&\le& \sum_{n \ge \mathcal{N}'} \frac{e^n}{(n-2)!}\left(\frac{\varepsilon^{1/2}\g^{-1/2} e^{K_5-K_2}K_1^{1/6}(M +2 \zeta c_1)^{1/2}K_6}{3K_3^{\frac{1}{6}}K_4^{\frac{1}{2}}}\right)^n
 \tilde{K}_3^{-2n+2}\tilde{K}_4^{-2n+2}
\\
 & & \times \sum_{\tau\in T_n}
 \sum_{\{i,j\}\in \tau} \sum_{s_{ij}=1}^{6}
 \prod_{k=1}^{n} \Gamma(d_k)
 [e\mathcal{O}(1) K]^{(n-1)}F_{zz'}^{(2/3)}
\\
&\le& \sum_{n \ge \mathcal{N}'} \frac{e^n}{(n-2)!}\left(\frac{\varepsilon^{1/2}\g^{-1/2} e^{K_5-K_2}K_1^{1/6}(M +2 \zeta c_1)^{1/2}K_6}{3K_3^{\frac{1}{6}}K_4^{\frac{1}{2}}}\right)^n
 \tilde{K}_3^{-2n+2}\tilde{K}_4^{-2n+2}
\\
 & & \times [e\mathcal{O}(1) K]^{(n-1)} F_{zz'}^{(2/3)} 6^{n-1}
 \sum_{d_1+\dots+d_n=2n-2\atop d_i \ge 1}
 \prod_{k=1}^{n} (d_k-1)!
 \frac{(n-2)!}{\prod_{i=1}^{n}(d_i-1)!}
\\
&\le& \sum_{n \ge \mathcal{N}'} \left(\frac{\varepsilon^{1/2}\g^{-1/2} e^{1+K_5-K_2}K_1^{1/6}(M +2 \zeta c_1)^{1/2}K_6}{3K_3^{\frac{1}{6}}K_4^{\frac{1}{2}}}\right)^n
 \tilde{K}_3^{-2n+2}\tilde{K}_4^{-2n+2}
\\
& & \times [e\mathcal{O}(1) K]^{(n-1)} F_{zz'}^{(2/3)} 6^{n-1}4^n,
\een
where we used the inequality
\be
 \sum_{d_1+\dots+d_n=2n-2\atop d_i \geq 1} 1 = \sum_{y_1+\dots+y_n=n-2\atop y_i \ge 0} 1 =
 \left( \begin{array}{c} 2n-3 \\ n-2 \end{array} \right) \leq 2^{2n-3} \leq 4^n ~,
\ee
with $y_i = d_i-1$.
Hence,
\be
 \sum_{R\subset \La: |R|\ge 2\atop z,z' \in R}|\rho(R)| e^{|R|} & \leq &
\frac{4\varepsilon^{1/2}\g^{-1/2} e^{1+K_5-K_2}K_1^{1/6}(M +2 \zeta c_1)^{1/2}K_6}{3K_3^{\frac{1}{6}}K_4^{\frac{1}{2}}} F_{zz'}^{(2/3)}
\\
& &\times \sum_{n \geq \mathcal{N}^{*}} \left\{\frac{8\varepsilon^{1/2}\g^{-1/2}e^{2+K_5-K_2}K_1^{1/6}(M +2 \zeta c_1)^{1/2}K_6\mathcal{O}(1)}{K_3^{\frac{1}{6}}K_4^{\frac{1}{2}}\tilde{K}_3^2 \tilde{K}_4^2}K \right\}^n
\nonumber\\
&=& c\left\{\frac{8\varepsilon^{1/2}\g^{-1/2}e^{2+K_5-K_2}K_1^{1/6}(M +2
\zeta c_1)^{1/2}K_6\mathcal{O}(1)}{K_3^{\frac{1}{6}}K_4^{\frac{1}{2}}\tilde{K}_3^2 \tilde{K}_4^2}K \right\}^{\mathcal{N}^{*}} F_{zz'}^{(2/3)}
\nonumber\\
& &\times \sum_{n = 0}^{\infty} \left\{\frac{8\varepsilon^{1/2}\g^{-1/2}e^{2+K_5-K_2}K_1^{1/6}(M +2 \zeta c_1)^{1/2}K_6\mathcal{O}(1)}{K_3^{\frac{1}{6}}K_4^{\frac{1}{2}}\tilde{K}_3^2 \tilde{K}_4^2}K \right\}^n  ,
\ee
where $\mathcal{N}^{*} = \mathcal{N}'-1$ and
\ben
c=\frac{4\varepsilon^{1/2}\g^{-1/2} e^{1+K_5-K_2}K_1^{1/6}(M +2 \zeta c_1)^{1/2}K_6}{3K_3^{\frac{1}{6}}K_4^{\frac{1}{2}}}.
\een

In short, we have proved the following result.\\

\begin{lemma}
\label{Lemma3}
If $K= \max\{A_1, A_2, \dots, A_6\}$ is sufficiently small,
then,
\ben
\e(K) = \frac{8\varepsilon^{1/2}\g^{-1/2}e^{2+K_5-K_2}K_1^{1/6}(M +2 \zeta c_1)^{1/2}K_6\mathcal{O}(1)}{K_3^{\frac{1}{6}}K_4^{\frac{1}{2}}\tilde{K}_3^2 \tilde{K}_4^2}K
\een
is a positive function and,
for any $z\in \La, z' \in \La$ with $z\neq z'$
\be
\label{somaR}
\sum_{R\subset \La: |R|\ge 2\atop z,z' \in R}|\rho(R)| e^{|R|}
\le c[\e(K)]^{\mathcal{N}^{*}} F_{zz'}^{(2/3)} = c[\e(K)]^{\max\{\frac{|z_0-{z'}_0|}{\varepsilon},1\}} F_{zz'}^{(2/3)}.
\ee
\end{lemma}

From the lemma \ref{Lemma3}, we obtain

\begin{corollary}
\label{Corol4}
\ben
\sup_{x\in \Z^{d+1}} \sum_{R:x\in R}|\rho(R)|e^{|R|}
\le c\mathcal{O}(1) \e(K).\\
\een
\end{corollary}
\begin{proof}
In fact, as $\rho(R)=0$ if $|R|=1$, we have
\ben
 \sup_{x\in \Z^{d+1}} \sum_{R:x\in R}|\rho(R)|e^{|R|} &=&
  \sup_{x\in \Z^{d+1}}\sum_{R:x\in R\atop |R|\ge 2}|\rho(R)|e^{|R|}
 \le \sup_{x\in \Z^{d+1}}\sum_{z\in\Z^{d+1}:\atop z\neq x}
   \sum_{R:|R|\ge 2\atop x,z \in R}|\rho(R)|e^{|R|}
\\
 &\le& \sup_{x\in \Z^{d+1}}\sum_{z\in\Z^{d+1}:z\neq x}
  c[\e(K)]^{\max\{\frac{|x_0-z_0|}{\varepsilon},1\}} F_{xz}^{(2/3)}
 \le c\mathcal{O}(1) \e(K),\\
\een
since, for $w>0$,
\ben
 \sum_{z\in\Z^{d+1}:z\neq x}
 e^{-w\frac{|x_0-z_0|}{\varepsilon}}\left[\frac{(1-\delta_{\vec{x},\vec{z}})}{|\vec{x}-\vec{z}|^p}
 + \delta_{\vec{x},\vec{z}}
 \right]
= \mathcal{O}(1),\\
\een
and so $\sum_{z\in\Z^{d+1}:z\neq x} F_{xz}^{(2/3)} \le \mathcal{O}(1)$.
\end{proof}

The results above establish the convergence of the cluster expansion
for $\e(K)$ small enough such that $c\mathcal{O}(1)\e(K) < 1$, i.e., for $\zeta$, $M$ and, mainly, $\lambda$ large.
See lemma \ref{lemma:KP} and Eqs.(\ref{cluster1}, \ref{cluster2}).


\section{Decay of two-point correlation} \label{sec:decay}

As it is well known, the convergence of the cluster expansion assures the
decay of the correlation functions and lead to direct estimates.
We present the main technical details related to the behavior of the
truncated two-point function below.

Turning to the expression
\be
\label{rhotilqp}
\tilde \r(R_i)=\!\!
\prod_{x\in R_i}\!\!\int d\n(\phi_x)
(q_{x_1}^{\b_i^{1}}+\b_i^{1} l_q)
(p_{x_2}^{\b_i^{2}}+\b_i^{2} l_p)
\sum\limits_{g\in G_{R}}
\prod\limits_{\{ x,y\}\in g}\!\!\!\!\!
(e^{  G_{xy}(\phi_x, \phi_y)}-1),
\ee
 which defines
$\tilde{\rho}(R_i)$, with $\b_i^j = 0$ if $i\neq i_j$, or $1$ if $i= i_j$,
$l_q=\int q d\n(\psi)$ and $l_p=\int p d\n(\psi)$, we note that the index $i$ of the term $\b^j_i$
is the same of the polymer $R_i$,
and so $i\in\{1,2,3,\dots,n\}$ as $j\in\{1,2\}$.
Consider the expression (\ref{SerieS2}) for $S_2(x_1;x_2)$,
and recall that $i_1,i_2 \in \{1,2,3,\dots,n\}$. Then, we have
two distinct cases: $i_1=i_2$ or $i_1\neq i_2$.
If $i_1=i_2$, then $\{x_1,x_2\} \subset R_{i_1}$,
$\{x_1,x_2\} \cap R_{i} = \emptyset~ \forall i\neq i_1$ and
$ \b_{i_1}^1=\b_{i_1}^1=\b_{i_2}^2=\b_{i_2}^2=1$.
If $i_1\neq i_2$, then $x_1 \in R_{i_1}$, $x_1 \notin R_{i_2}$,
$x_2 \in R_{i_2}$, $x_2 \notin R_{i_1}$,
$\{x_1,x_2\} \cap R_{i} =\emptyset~ \forall i \notin \{i_1,i_2\}$,
$ \b_{i_1}^1=\b_{i_2}^2=1$, and $\b_{i_1}^2=\b_{i_2}^1=0$.

Hence, as $1=(1-\d_{i_1,i_2})+\d_{i_1,i_2}$ we rewrite
\be
S_2(x_1;x_2) = D_1(x_1,x_2) + D_2(x_1,x_2),
\ee
where
\be
D_1(x_1,x_2) &\equiv& \sum_{n\geq 1}{1\over n!}
\sum_{i_1,i_2 = 1}^{n} (1-\d_{i_1,i_2})\!\!\!\!\!\!\!\!\!\!\!\!\!\!\!
\sum_{R_1, \dots, R_n\subset\La,~|R_j|\ge 2\atop
R_{i_1}\ni x_1 R_{i_2}\ni x_2}\!\!\!\!\!\!\!\!\!\!\!\!\!\!
\phi^T(R_1 ,\dots, R_n)
\tilde \rho(R_1)\dots \tilde \rho(R_n)
\nonumber \\
&=& \sum_{n\geq 2}{1\over (n-2)!}\!\!\!
\sum_{R_1, \dots, R_n\subset\La,~|R_j|\ge 2\atop
R_1\ni x_1 R_2\ni x_2}\!\!\!\!\!\!\!\!\!\!\!\!\!\!\!
\phi^T(R_1 ,\dots, R_n)
\tilde \rho(R_1)\dots \tilde \rho(R_n),
\\
D_2(x_1,x_2) &\equiv& \sum_{n\geq 1}{1\over n!}
\sum_{i_1,i_2 = 1}^{n} \d_{i_1,i_2}\!\!\!\!\!
\sum_{R_1, \dots, R_n\subset\La,~|R_j|\ge 2\atop
R_{i_1}\ni x_1 R_{i_2}\ni x_2}\!\!\!\!\!\!\!\!\!\!\!\!\!\!\!
\phi^T(R_1 ,\dots, R_n)
\tilde \rho(R_1)\dots \tilde \rho(R_n)
\nonumber \\
&=& \sum_{n\geq 1}{1\over (n-1)!}
\sum_{R_1, \dots, R_n\subset\La,~|R_j|\ge 2\atop
R_1 \supset\{x_1,x_2\}}\!\!\!\!\!\!\!\!\!\!\!\!\!\!\!\!\!
\phi^T(R_1 ,\dots, R_n)
\tilde \rho(R_1)\dots \tilde \rho(R_n),
\ee
since, in $D_1(x_1,x_2)$ when $n=1$ we have
$\sum_{i_1,i_2=1}^{1}(1-\d_{i_1,i_2})=0$ and, for any $n>2$ the sum
$\sum_{i_1,i_2=1}^{n}(1-\d_{i_1,i_2})$ leads to $n(n-1)$ equal
terms. And, in $D_2(x_1,x_2)$ the sum
$\sum_{i_1,i_2=1}^{n}\d_{i_1,i_2}$ gives $n$ equal terms.

Thus,
\ben
|S_2(x_1;x_2)|\le |D_1(x_1,x_2)|+ |D_2(x_1,x_2)|.
\een

Comparing (\ref{rhotilqp}) with (\ref{rhoR}), we note that if
$R_i\cap\{x_1,x_2\}=\emptyset$ then $\tilde\rho(R_i)=\rho(R_i)$. If
$R_i\cap\{x_1,x_2\}\neq\emptyset$, we can obtain the result
(\ref{somaR}) of the lemma \ref{Lemma3} for $\tilde\rho(R_i)$ by
changing $n_k(s)$ and $m_k(s)$ by $n_k(s)+1$ and $m_k(s)+1$. With such result, we change
$\Gamma(d_k)$ by $\Gamma(d_k+1)$ in (\ref{Soma1}) and obtain an
extra $\prod_{k=1}^{n}d_k$, which is bounded by $e^{2(n-1)}$. We use
the lemma \ref{Cota da integral} with $\alpha=\beta=1$ to bound the factors $l_q$ and $l_p$
in (\ref{rhotil}). Hence, we can apply the lemma \ref{Lemma3} and
corollary \ref{Corol4} to estimate $\tilde\rho(R_i)$ (changing some
multiplicative constants).

Now, let us find an upper bound for  the term $|D_1(x_1,x_2)|$. We have
\be
| D_1(x_1,x_2)|
\leq
\sum_{n\geq 2}{1\over (n-2)!}
B_{n}(x_1,x_2),
\ee
where
\ben
B_{n}(x_1,x_2)=\!\!\!\!\!\!
\sum_{R_1,\dots ,R_{n}\subset\La\atop
|R_{i}|\ge 2,\,
x_1\in R_1,x_2\in R_2}\!\!\!\!\!\!\!\!\!\!\!\!\!\!\!\!
\left|\phi^{T}(R_1 ,R_2 ,\dots , R_n)
|\tilde\r(R_1)|
|\tilde\r(R_2)| |\tilde\r(R_3)\dots\tilde\r(R_n)\right|.
\een
Note that in (\ref{phiT}), for $n\geq 2$,
$\phi^{T}(R_1 ,\dots , R_n)>0$ only if $g(R_1 ,\dots ,R_n)\in G_n$.
Thus,
\ben
\sum_{R_1,\dots ,R_{n}\subset\La\atop
|R_{i}|\ge 2,\,
x_1\in R_1,x_2\in R_2}\!\!\!\!\!\!\!\!\!\!\!\!\!\!
\left|\phi^{T}(R_1 ,R_2 ,\dots , R_n)\right| [\cdot] =
\sum_{g\in G_n}\left|\sum_{f\in G_n\atop\subset g}(-1)^{|f|}\right|
\sum_{R_{1},\dots ,R_{n}\subset\La:\,
|R_{i}|\ge 2\atop g(R_{1},\dots ,R_{n})=g,
\,x_1\in R_1,x_2\in R_2}\!\!\!\!\!\!\!\!\!\!\!\!\!\!\!\!\!\!\!\!
[\cdot].
\een
By the Rota formula \cite{Rota}, we have
\be
\left|\sum_{f\in G_n\atop f\subset g}(-1)^{|f|}\right|\leq
\sum_{\tau\in T_n: \tau \subset g}1 \equiv N(g).
\ee
A proof of the Rota formula above can be found e.g.
in Refs. \cite{Rota} and \cite{Si}.

We recall now that \ben \sum_{g\in G_n}[\cdot]=\sum_{\t\in T_n}
\sum_{g: \,\t\subset g}{1\over N(g)}[\cdot], \een since in the
double sum $\sum_{\t}\sum_{g\supset t}$ each $g$ will be repeated
exactly $N(g)$ times.

Thus,
\ben
B_{n}(x_1,x_2)\leq
\sum_{\t\in T_n}w_n(\t,x_1,x_2),
\een
where
\ben
w_n(\t,x_1,x_2)\equiv
\sum_{R_{1},\dots ,R_{n}\subset\La:\,
|R_{i}|\ge 2\atop g(R_1,R_{2},\dots ,R_{n})\supset \t,
\,x_1\in R_1,x_2\in R_2}\!\!\!\!\!\!\!\!\!\!\!\!\!\!\!\!\!\!\!\!\!\!\!\!
|\tilde\r(R_1)|
|\tilde\r(R_2)||\tilde\r(R_3)\dots\tilde\r(R_n)|.
\een

Using now the obvious bound
\ben
\sum_{R:\, R\cap R'\neq\emptyset}|\cdot|\leq
|R'|\sup_{x\in R'}\sum_{R:\, x\in R}|\cdot|,
\een
and denoting again as $\bar\t$ the subtree of $\t$ which is the unique path
joining vertex $1$ to vertex $2$, and denoting
as $I_\t=\{1,i_1,\dots ,i_k,2\}$ the
ordered set of
the vertices of $\bar\t$,
one can  easily check that
\ben
& &w_n(\t,x_1,x_2)\leq
\prod_{i\notin I_\t}^{n}\left[\sup_{x\in \mathbb{Z}}\sum_{R_i:\, x\in R_i}
|R_i|^{d_i -1}|\tilde\r(R_i)|\right]
\\
& & \times
\sum_{R_1,R_{i_1},\dots, R_{i_k},R_2:x_1\in R_1 ,x_2\in R_2
\atop R_1\cap R_{i_1}\neq \emptyset,\dots R_{i_k}\cap R_2\neq \emptyset }\!\!\!\!\!\!\!\!\!\!\!\!\!\!\!\!\!\!
|R_1|^{d_1-1}|\tilde\r(R_1)|
|R_2|^{d_2-1}|\tilde\r(R_2)|\prod_{i\in I_\t\atop i\neq 1,2}^{n}
|R_i|^{d_i-2}|\tilde\r(R_i)|
\\
&\leq&
\prod_{i\notin I_\t}^{n}\left[\sup_{x\in \L}\sum_{R_i:\, x\in R_i}
(d_i-1)!|\tilde\r(R_i)|e^{|R_i|}\right](d_1-1)!(d_2-1)!
\\
& &\times
\sum_{R_1,R_{i_1},\dots, R_{i_k},R_2:x_1\in R_1 ,x_2 \in R_2 \atop
R_1\cap R_{i_1}\neq \emptyset,
\dots R_{i_k}\cap R_2\neq \emptyset}\!\!\!\!\!\!\!\!\!\!\!\!\!\!\!\!\!\!\!\!\!
|\tilde\r(R_1)|e^{|R_1|}
|\tilde\r(R_2)|
e^{|R_2|}\prod_{i\in I_\t\atop i\neq 1,2}{(d_i -2)!}
|\tilde\r(R_i)|e^{|R_i|},
\een
since $|R|^n \leq n! e^{|R|}$.
Now, note that
\ben
 \sum_{R_1,R_{i_1},\dots, R_{i_k},R_2:
x_1\in R_1 ,x_2\in R_2\atop
R_1\cap R_{i_1}\neq \emptyset,\dots
R_{i_k}\cap R_2\neq \emptyset} \leq
 \sum_{x_{i_0}\in \L}\sum_{x_{i_1}\in \L}\dots
\sum_{x_{i_k}\in \L}
\sum_{R_{1}\atop x_1,x_{i_0}\in R_{1}}
\sum_{R_{i_1}\atop x_{i_0},x_{i_1}\in R_{i_1}}
\!\!\dots\!\!
\sum_{R_{i_k}\atop x_{i_{k-1}},x_{i_k}\in
R_{i_k}}
\sum_{R_{2}\atop x_{i_{k}},x_2\in R_{2}}.
\een

Hence, recalling Eq.(\ref{somaR}) and applying iteratively the
inequality (\ref{SomaFxy}) with $w_1=1/2$ and $w_2=2/3$,
\ben
& &\sum_{R_1,R_{i_1},\dots, R_{i_k},R_2:
x_1\in R_1 ,x_2\in R_2\atop
R_1\cap R_{i_1}\neq \emptyset,\dots
R_{i_k}\cap R_2\neq \emptyset}\!\!\!\!\!\!\!\!\!\!\!\!\!\!\!\!\!\!\!\!\!
|\tilde\r(R_1)|e^{|R_1|}
|\tilde\r(R_2)|
e^{|R_2|}
\prod_{i\in I_\t}
|\tilde\r(R_i)|e^{|R_i|}\\
&\leq&
\sum_{x_{i_0}\in \L}\sum_{x_{i_1}\in \L}\dots
\sum_{x_{i_k}\in \L}
c [\e(K)]^{{\rm max}\{\frac{|(x_1)_0-(x_{i_0})_0|}{\varepsilon},1\}}F_{x_1x_{i_0}}^{(2/3)}
\dots
c [\e(K)]^{{\rm max}\{\frac{|(x_{i_k})_0-(x_{2})_0|}{\varepsilon},1\}}F_{x_{i_k}x_{2}}^{(2/3)}\\
&\leq&
[\mathcal{O}(1)]^{k+1}c^{k+2}[\e(K)]^{{\rm max}\{\frac{|(x_{1})_0-(x_{2})_0|}{\varepsilon},k+2\}}
F_{x_1 x_2}^{(1/2)},
\een
since $\e(K) < 1$ and 
$|(x_1)_0-(x_2)_0| \le |(x_1)_0-(x_{i_0})_0|+|(x_{i_0})_0-(x_{i_1})_0|
+\dots+|(x_{i_k})_0-(x_{2})_0|$.\\

Thus, using the corollary \ref{Corol4} and noting that
$|\{1,...,n\}\backslash I_\t|=n-k-2$,
\ben
w_n(\t,x_1,x_2)&\leq&
(d_1-1)!(d_2-1)!
\left[\prod_{i\notin I_\t}^{n}\sup_{x\in \mathbb{Z}}\sum_{R_i:\, x\in R_i}
(d_i-1)!|\r(R_i)|e^{|R_i|}\right]\\
& \times &
\left[\prod_{i\in I_\t\atop i\neq 1,2}^{n}
(d_i -2)!\right]
c^{k+2}[\e(K)]^{{\rm max}\{\frac{|(x_1)_0-(x_2)_0|}{\varepsilon}, k+2\}} [\mathcal{O}(1)]^{k+1}
F_{x_1 x_2}^{(1/2)}\\
&\leq&
[\mathcal{O}(1)]^n c^n [\e(K)]^{{\rm max}\{\frac{|(x_1)_0-(x_2)_0|}{\varepsilon}, n\}} F_{x_1 x_2}^{(1/2)}
\prod_{i=1}^{n}{(d_i -1)!}.
\een
Finally, carrying out the sum over $\t$ (and using, once again, the Cayley
formula) we obtain
\ben
B_n(x_1,x_2)\leq (n-2)![4\mathcal{O}(1)]^n[\e(K)]^{{\rm max}\{\frac{|(x_1)_0-(x_2)_0|}{\varepsilon}, n\}}
F_{x_1 x_2}^{(1/2)}.
\een

Taking $K$ small enough to make
$4\mathcal{O}(1)\e(K)<1$,
for the contribution of $D_1$ to the correlations, we get the following
bound:
\be
| D_1(x_1,x_2)|
\leq
\sum_{n\geq 2}[4\mathcal{O}(1)]^n[\e(K)]^{{\rm max}\{\frac{|(x_1)_0-(x_2)_0|}{\varepsilon}, n\}}
 F_{x_1 x_2}^{(1/2)}
\leq \mathcal{O}(1) [\e(K)]^{\frac{|(x_1)_0-(x_2)_0|}{\varepsilon}} F_{x_1 x_2}^{(1/2)}.
\ee
In a similar and much
easier way one can also prove a completely analogous
bound for $|D_2(x_1,x_2)|$
\be
|D_2(x_1,x_2)|\le \mathcal{O}(1) [\e(K)]^{\frac{|(x_1)_0-(x_2)_0|}{\varepsilon}}
F_{x_1 x_2}^{(1/2)}.
\ee

Hence,
\be
|S_2(x;y)|
&\le& \mathcal{O}(1) [\e(K)]^{\frac{|x_0-y_0|}{\varepsilon}} F_{x y}^{(1/2)}
\le \mathcal{O}(1) [\e(K)]^{\frac{|x_0-y_0|}{\varepsilon}}
 e^{-\frac{|x_0-y_0|}{2\varepsilon}}\left(\frac{1-\d_{\vec{x},\vec{y}}}{|\vec{x}-\vec{y}|^p}
 +\d_{\vec{x},\vec{y}}\right) \nonumber \\
&\le& \mathcal{O}(1) e^{-m'(K)|x_0-y_0|}
\left(\frac{1-\d_{\vec{x},\vec{y}}}{|\vec{x}-\vec{y}|^p}
 +\d_{\vec{x},\vec{y}}\right),
\ee
where, since $\e(K) < 1$, we write above
\ben
m'(K) \equiv \frac{-\log[\e(K)] + 1/2}{\e} > 0.
\een

It is important to remark that the existence of a convergent polymer expansion, such as that presented above, allows us to obtain
also a lower bound for the correlations. Roughly, if we write the polymer series as a main term plus corrections, we get the upper bound;
and the lower bound is given by the main term minus corrections, see Ref.\cite{PS}.

In short, the results of this section may be summarized as follows.

\begin{theorem}
The two-point function $S_{2}(x;y)$ (\ref{principal}) of the anharmonic chain of oscillators with discrete times, written as a polymer expansion (\ref{SerieS2}), converges absolutely and uniformly  in the
volume $|\Lambda|$ (number of sites $N$ and time $\mathfrak{T}$), for $\zeta, M, \lambda$ large enough. Moreover, $S_{2}(x;y)$ has the  upper bound
\begin{equation*}
|S_2(x;y)|
\leq C' e^{-m'|x_0-y_0|}\left(\frac{1-\d_{\vec{x},\vec{y}}}{|\vec{x}-\vec{y}|^p}
+\d_{\vec{x},\vec{y}}\right).
\end{equation*}
And a similar lower bound follows, with other properly chosen parameters $C''$ and $m''$.
\end{theorem}

Some short notes are appropriate here.

As described above, the decay in space of the two-point function $S_{2}$ is polynomial, and follows the decay of the interparticle interaction $J_{j,\ell}$.
As the two-point function is directly related to the heat flow and to the thermal conductivity in these systems given by chains of oscillators, such a result
is of direct interest: see e.g. Ref.\cite{PA} in which we assumed the space decay of terms in the expression of the heat flow related to the space decay of the
interparticle interaction - result which is proved by theorem 2.

An investigation about the precise rate for the exponential decay in time of $S_{2}$ (something between $m'$ and $m''$) may be possible by using standard techniques of constructive field
theory related to spectral analysis \cite{GJ}, but it is beyond the aim of the present work. See e.g. Refs. \cite{P2002,FOPS} for examples of detailed study of the two and four-point
correlations decay in time, via such an approach, in the stochastic Ginzburg-Landau model (a simpler system with nonconservative dynamics and, in these specific works, relaxing to equilibrium).

Finally, we have a remark about the time discretization. Here, as previously described, we work on a lattice with time step $\varepsilon$. To follow the dependence on $\varepsilon$, note that  such factors
 are hidden in some terms, for example, in $|x_{0}-y_{0}|$ (which is a multiple of $\varepsilon$). Hence, in $m'|x_{0}-y_{0}|$, which appears in the theorem above, there is a factor $\varepsilon$
in the denominator of the expression for $m'$, as well as another one in the numerator within $|x_{0}-y_{0}|$. However, we need to say that, considering the whole problem (all expressions and manipulations), if we naively try
to recover the original nonlinear model with continuous time by simply taking the limit of $\varepsilon$ going to zero, divergences and problems will appear. In short, recovering the continuous limit is not a
trivial work. To illustrate such adversity, we note that the expression for the Gaussian measure related
to the harmonic part (see Eq.(\ref{quadratica})) is well defined and controllable in the continuous limit, but we do not have the inverse of the covariance $\mathcal{C}$, i.e., the diffusion matrix $\mathcal{C}$ is not invertible in the continuous time limit (and, we recall, $\mathcal{C}^{-1}$, or a related expression, is well defined and important in the formalism with discrete time). Such trouble  is not specific for the present investigation: it is very well known in the study of stochastic processes, see e.g. Ref.\cite{Drozdov} and references there in. Moreover, such difficulty in taking the continuous limit is also common in other related problems in
physics, as already said: recall, for example, the ultraviolet (UV) limit in Quantum Field Theory.  Anyway, in a lattice with a fixed step, we can still obtain a precise description for the heat flow investigation, as we confirm in the next section.


\section{A Concrete Example: Analysis of the Anharmonic Chain with Quartic On-Site Potential}\label{sec:exemplo}

Now we turn to the analysis of a concrete and recurrent problem: the heat flow in the anharmonic chain of oscillators with quartic on-site potential. Here, besides the specific quartic on-site potential, we take a
 model with weak and nearest-neighbor interparticle interactions, and with the inner reservoirs in the self-consistent condition. Precisely, for the interpaticle interaction we take  $\mathcal{J}_{ij}\neq 0 \iff i = (j+N)\pm 1$, $|\mathcal{J}_{ij}|\ll 1$; and, for the anharmonic on-site potential, $\mathcal{P}(\phi_{j}) =
\phi_{j}^{4}/4$. Moreover, for simplicity, besides the regimes already considered ($m=1$, $M=2+\varepsilon$), we still assume the regime of
high anharmonicity and temperature.

The present section is directed toward a
twofold aim: first, to show the usefulness of the polymer expansion convergence by describing an interesting result obtained within a perturbative analysis; second, to show the trustworthiness of the discrete time approximation by
presenting results which, in comparison with well known numerical computations, are precise.

We need to remark that quite similar results have been already described, by some of the authors, in a previous work \cite{PFL}. However, at that time, the analysis was carried out in an uncontrolled perturbative approach: the convergence of a related
polymer expansion was unknown. For completeness, we repeat some details here.

The heat current, as previously described, follows from Eq.(\ref{fluxo1})
$$
\mathcal{F}_{j,j+1} = \frac{J_{j,j+1}}{2} \Big\langle
 (\varphi_{j} - \varphi_{j+1})(\varphi_{j+N} + \varphi_{j+1+N})
\Big\rangle .
$$
To analyze it in the steady state, we need to study the averages of $\varphi_{j}(\mathfrak{T})\varphi_{j+N+1}(\mathfrak{T})$,
$\varphi_{j}(\mathfrak{T})\varphi_{j+N}(\mathfrak{T})$, etc., as $\mathfrak{T}\rightarrow\infty$. That is, we need to evaluate some two-point correlation functions.
Recall that, from Eq.(\ref{correlations}), the correlations were first written as
$$
\left\langle \varphi_{i}(\mathfrak{T})\varphi_{j}(\mathfrak{T}) \right\rangle =
\int \phi_{i}(\mathfrak{T})\phi_{j}(\mathfrak{T}) \exp[-W(\phi)] d\mu_{\mathcal{C}}\ ,
$$
where
\begin{eqnarray*}
W(\phi) &=& \int_{0}^{\mathfrak{T}} \!\!\!\! \phi_j(s)\mathcal{J}^{\dagger}_{ji}\gamma_i^{-1}d\phi_i(s) +
\lambda\gamma_i^{-1}{P}'(\phi)_i(t)d\phi_i(s) +
 \phi_j(s)\mathcal{J}_{ij}^{\dagger}\gamma_{i}^{-1}A^{0}_{ik}\phi_k(s) ds + \\
&& + \lambda
\gamma_{i}^{-1} {P}' (\phi)_i(s) A^0_{ik} \phi_k(s) ds +
\frac{1}{2}\phi_{j'}(s)\mathcal{J}_{j'i}^{\dagger}\gamma_i^{-1}\mathcal{J}_{ij}\phi_j(s) ds +
\\
&& + \frac{1}{2}\lambda^2\gamma_i^{-1}({P}'(\phi)_i)^2(s) ds + \lambda \gamma_{i}^{-1}
{P}' (\phi)_i(s) \mathcal{J}_{ij}\phi_j(s) ds,
\end{eqnarray*}
and the Gaussian measure $d\mu_{\mathcal{C}}$ was given by Eq.(\ref{quadratica}). Due to exceeding difficulties in the investigation within this first formalism, our idea to perform the computation  is resumed in the following strategy (as exhaustively emphasized throughout the paper): we
introduce the approximation of discrete times, see Eq.(\ref{principal}), and rearrange the integral representation in terms of a new measure, precisely, a suitable single spin distribution (SSD) with nonlinear parts, as
presented in section The Polymer Expansion. In short, we rewrite $\exp[-W(\phi)] d\mu_{\mathcal{C}}$ as $\exp[-\tilde{W}(\phi)]
d \nu$.
In this properly built SSD $d\nu$, instead of considering the fields  $\phi_j$ and $\phi_i$ always in separate (as in a usual polymer expansion), we join in the same cell the pairs $\phi_{j}(s)$ and
$\phi_{i}(s+\varepsilon)$ with $i=j+N$ (of course, $\phi_{j}(\mathfrak{T})$ and $\phi_{i}(0)$ do not have pairs). Precisely, our SSD is given by the expression
\begin{widetext}
\begin{equation}
d\nu(\phi_{j}(s),\phi_{i=j+N}(s+\varepsilon)) = \exp\left\{\varepsilon\left[ -\frac{1}{2}\lambda^{2}\gamma_{j}^{-1}\phi_{j}^{6}(s)
  - \frac{1}{2T_{i}}\phi_{i}^{2}(s+\varepsilon) - \gamma_{j}^{-1}\lambda\phi_{j}^{3}(s)\phi_{i}(s+\varepsilon) + \ldots \right]\right\} d\phi_{j}(s) d\phi_{i}(s+\varepsilon) / N ,
\end{equation}
\end{widetext}
where the dots above describe subdominant terms, $N$ is the normalization, and $\phi_{i}^{2}(s+\varepsilon)$ was extracted from $(\phi,\mathcal{D}^{-1}\phi)$, which comes from the harmonic potential related to
the Gaussian measure which appeared in the previous formalism. And $\tilde{W}(\phi)$ above is given by subdominant terms that we left behind, both from $\exp[-W(\phi)]$ and $d\mu_{\mathcal{C}}$, after writing the expression for the SSD, i.e.
\begin{multline}
\tilde{W}(\phi) =
 -\sum_{s, i,j,\ldots} \varepsilon\left[
  \phi_{j}(s)\mathcal{J}_{ji}^{\dagger}\gamma_{i}^{-1}\phi_{i}(s+\varepsilon) + \right. \phi_{j}(s) \mathcal{J}_{ji}^{\dagger}\frac{M_{i-N}}{\gamma_{i}}\phi_{i-N}(s) \\
+ \frac{1}{2\gamma_{i}}\phi_{j'}(s) \mathcal{J}_{j'i}^{\dagger}\mathcal{J}_{ij}\phi_{j}(s)
+ \left.  \frac{\lambda}{\gamma_{i}} {P}'(\phi_{i-N}(s))\mathcal{J}_{ij}\phi_{j}(s) +
\frac{1}{2}\phi_{k}(s) \tilde{\mathcal{D}}_{k,k'}^{-1}(s,s') \phi_{k'}(s')
\right],
\end{multline}
where $\tilde{\mathcal{D}}_{k,k'}^{-1}$ is the quadratic part $\mathcal{D}_{k,k'}^{-1}$ without the terms which are already considered in the SSD.

Note that, essentially, $\phi_{i}^{2}$ and $\phi_{j}^{6}$ rule  the behavior of the SSD above, and the forthcoming computations.
Hence, in our final formalism here, the integral representation for the two-point function is given as product of these SSD (with cells of sites $[j,s]$
and $[i=j+N,s+\varepsilon]$) and the exponential of
terms involving the weak interaction $J$, which couples different cells, and the remaining terms from $(\phi,\mathcal{C}^{-1}\phi)$,
which are also small: e.g., for the part involving $\phi_{j}$, in the regime of large anharmonicity, rescaling the dominant term
$\lambda^{2}\phi_{j}^{6}$ as $\tilde{\phi}_{j}^{6}$ in the s.s.d., this part will involve $\tilde{\phi}_{j}$ and powers of $1/\lambda$.

Now, we perform a perturbative computation, considering in $\exp\{-\tilde{W}(\phi)\}$ only the terms up to first order, i.e.
taking $\exp[-\tilde{W}(\phi)] \approx 1 -\tilde{W}(\phi)$. These leading terms are directly related to the terms which appear in the polymer expansion
written as $(e^{G_{x,y}} - 1)$; see section The Polymer Expansion.

Thus,  carrying out the computations, we note that a first important contribution is given by
\begin{align*}
\lefteqn{\int (\phi_{i-N}(\mathfrak{T})\phi_{i+1}(\mathfrak{T}))\cdot \varepsilon[\lambda\gamma_{i+1}^{-1}\phi_{i+1-N}^{3}(\mathfrak{T}-\varepsilon)\phi_{i+1}(\mathfrak{T})]_{*}} \\
& \cdot \varepsilon[\lambda\gamma_{i+1}^{-1}\phi_{i+1-N}^{3}(\mathfrak{T}-\varepsilon)\mathcal{J}_{i+1,i-N}\phi_{i-N}(\mathfrak{T}-\varepsilon)] \\
& \cdot \varepsilon[\phi_{i-N}(\mathfrak{T}-\varepsilon)\mathcal{C}^{-1}_{i-N,i-N}(\mathfrak{T}-\varepsilon,\mathfrak{T})\phi_{i-N}(\mathfrak{T})] d\tilde{\nu}(\phi)\\
& \sim c'(\varepsilon) J \frac{1}{\lambda^{4/3}}\frac{T_{i+1}^{2/3}}{T_{i}}  ,
\end{align*}
where  $[\cdot]_{*}$ above comes from the ``cross'' term in the s.s.d.; $d\tilde{\nu}$ is the main part of the s.s.d. (involving
$\phi_{i}^{2}$ and $\phi_{j}^{6}$); and $c'$ is a numerical factor. And, a second important contribution comes from terms similar to
\begin{align*}
\lefteqn{\int (\phi_{i-N}(\mathfrak{T})\phi_{i+1}(\mathfrak{T}))\cdot \varepsilon[\phi_{i-N}(\mathfrak{T}-\varepsilon)\mathcal{J}^{\dagger}_{i-N,i+1}\gamma_{i+1}^{-1}\phi_{i+1}(\mathfrak{T})]} \\
& \cdot \varepsilon[\phi_{i-N}(\mathfrak{T}-\varepsilon)\mathcal{C}^{-1}_{i-N,i-N}(\mathfrak{T}-\varepsilon,\mathfrak{T})\phi_{i-N}(\mathfrak{T})] d\tilde{\nu}(\phi)\\
& \sim c'' J \frac{1}{\lambda^{4/3}}\frac{1}{T_{i}^{1/3}}  .
\end{align*}
 Hence, summing up  all leading terms (with $\mathfrak{T}\rightarrow\infty$), and considering a small difference
between $T_{i+1}$ and $T_{i}$ (such that $T_{i+1}^{\alpha}-T_{i}^{\alpha} \approx \alpha T_{i}^{\alpha-1}(T_{i+1}-T_{i})$), we get
\begin{equation}
\mathcal{F}_{j,j+1} \approx - c \frac{J^{2}}{\lambda^{4/3}}\frac{1}{T_{j}^{4/3}}(T_{j+1}-T_{j}) .
\label{fluxofinal}
\end{equation}

Now, after obtaining the expression for $\mathcal{F}_{j,j+1}$, the computation of the heat current in terms of the temperatures at the boundaries is straightforward.
The self-consistent condition in the steady state says that there is no neat heat flows from the inner reservoirs to the system, i.e., it gives
\begin{equation}
\mathcal{F}_{1,2} = \mathcal{F}_{2,3} = \ldots = \mathcal{F}_{N-1,N} \equiv \mathcal{F} \label{correntesteady}.
\end{equation}
These equations together with Eq.(\ref{fluxofinal}) give us
\begin{eqnarray*}
\mathcal{F}(\mathcal{C}T_{1}^{\alpha}) &=& T_{1} - T_{2} \\
\mathcal{F}(\mathcal{C}T_{2}^{\alpha}) &=& T_{2} - T_{3} \\
\ldots &=& \ldots \\
\mathcal{F}(\mathcal{C}T_{N-1}^{\alpha}) &=& T_{N-1} - T_{N} .
\end{eqnarray*}
Summing up the expressions above, we find
\begin{equation*}
\mathcal{F} = \mathcal{K}\frac{(T_{1}-T_{N})}{N-1} ,
\end{equation*}
where
\begin{equation*}
\mathcal{K} = \left\{ \mathcal{C}T_{1}^{\alpha} +
\mathcal{C}T_{2}^{\alpha} + \ldots + \mathcal{C}T_{N-1}^{\alpha}\right\}^{-1}\cdot (N-1)
,\label{condutividade}
\end{equation*}
with $\mathcal{C}^{-1} = \displaystyle{c(\varepsilon) \frac{J^{2}}{\lambda^{4/3}}}$, $\alpha=4/3$.
For a small gradient of temperature, i.e., if $T_{j} \approx T$, we have the Fourier's law in the chain with thermal conductivity
\begin{equation}
\mathcal{K} \sim c(\varepsilon) \frac{J^{2}}{\lambda^{4/3}T^{4/3}} .
\end{equation}

  We emphasize the
  interest of such result. For the case of a strong anharmonic on-site potential, the effects of the internal reservoirs become less important, and one expects a system with behavior close to that observed in a chain with thermal baths only at the boundaries. Exhaustive computer simulations have been already carried out for these anharmonic chains with quartic potential and thermal baths only at the ends, and they give a thermal conductivity $\mathcal{K} \approx 1/T^{1,35}$ \cite{AK1,AK2}, essentially the same result obtained by the perturbative computation within our approach with the approximation of discrete times. Similar results, still for the anharmonic chain with reservoirs only at the boundaries, are presented in Ref.\cite{LiLi}, here obtained by using nonequilibrium molecular dynamics. It would be very interesting to make a comparison between our findings and numerical simulations in the original  anharmonic model with inner noises, but we do not know any numerical result in the literature for such model, and we have
  to leave this work for the experts in computer techniques.


\section{Final remarks} \label{sec:final}

We conclude with some comments and remarks.

First, as example of the trustworthiness of the perturbative analysis within the integral formalism, we recall that, when restricted to the easier case of harmonic chain of  oscillators with the self-consistent condition, perturbative computations within our integral representation, with continue time and without simplification
in the quadratic term (see Ref.\cite{PF}), reproduces the well known result that is described in Ref.\cite{BLL}, there obtained by a completely different method. That is, within our integral approach, we have proved that Fourier's law holds in the chain of harmonic oscillators with self-consistent inner baths, and we
obtained the expression for the thermal conductivity as that derived by other methods.

The behavior of the two-point function for a system with interparticle interaction with polynomial decay, what is proved here, allows us to establish the heat flow between different sites in a chain
with interparticle interaction beyond nearest-neighbor sites. As already said, that is a problem of physical interest: for example, in a recent work \cite{PA}, one of the authors and a collaborator, by assuming the heat flow behavior (correlation decay in space), which is proved in the present paper, show that the existence of interparticle interactions beyond nearest neighbors may increase by thousand times the thermal rectification in a graded chain, and may also avoid the decay of such rectification with the system size, which are important properties for the theoretical study and even experimental fabrication of thermal diodes. The effect of long range interactions increasing the thermal rectification of anharmonic crystals has been confirmed, by computer simulations, even in
anharmonic crystals without inner self-consistent baths \cite{Cas}.

It is also worth to mention that, in a previous work within a perturbative computation (which could not be rigorously justified at that time, before the present
results), we describe nontrivial properties of the heat flow in an inhomogeneous anharmonic chain, namely, thermal rectification and negative differential
thermal resistance \cite{Prapid}.

To conclude, we believe that, even though within an effective model (anharmonic chain of oscillators with inner stochastic reservoirs) and an approximation in the integral representation (discrete times),
the approach and results presented here may be of great utility in the qualitative understanding of the heat flow properties in the steady state of high anharmonic systems submitted to different temperatures.
In particular, the existence of a convergent polymer expansion
making possible  a perturbative investigation is of usefulness in the study of systems with long range interaction, and also in the analysis of inhomogeneous and asymmetric models, in which, the important phenomenon of thermal rectification appears.

\section*{Acknowledgements}
We thank the referees for carefully reading the manuscript, and for the list of suggestions which helped us to improve the presentation of the paper. This work was partially supported by
CNPq, Brazil.


\appendix \section{On the used approximations} \label{apendice}

As an argument, beyond the  technical reasons already mentioned, to support the study of the system with time regularization (i.e., without short times),
we recall that most of the physical research problems related to similar models involve questions about properties of the steady state, reached as $\mathfrak{T} \rightarrow \infty$. Moreover, from numerical studies of
similar dynamical problems carried out by physicists \cite{Delfini}, lower frequencies seem to dominate the transport on a large scale (the scale of the whole chain).
That is, it seems that the time regularization does not spoil the main features of the original problem, related to heat flow properties in the steady state.

In relation to the other approximation assumed in the present paper, namely, the modified covariance (\ref{quadratica}), we show below that it is, indeed, the main part to the original harmonic interaction.

The expressions  for the covariance are given by Eqs.(\ref{covariance1},\ref{covariance}). First,  we note that the replacement
of $C(t,t)$ by
$C\equiv C(\infty,\infty)$
does not  change  the steady heat current for the harmonic case, as shown in Ref.\cite{PF}. And so, we study the covariance with such a replacement. Furthermore, we
may write $\exp[-(t-s)A_{0}]$, $t\geq s$, as
\begin{eqnarray}
\exp\left (-|\tau| A^0 \right ) = ~ e^{-|\tau| \alpha}\left[ \cosh(\tau\rho)
\left (\begin{array}{cc}1 & 0 \\0&1 \end{array}\right )
+\frac{\sinh(|\tau|\rho)}{\rho}\left (
\begin{array}{cc}
\alpha & I \\
-\mathcal{M} & -\alpha
\end{array} \right ) \right],\label{exp}
\end{eqnarray}
where $\tau=t-s$, $\alpha=\zeta/2$, $\rho=\left(\alpha^2-M\right)^{1/2}$; and a similar expression given by the transposed matrix follows for negative $\tau$, see  Ref.\cite{BLL}. To proceed, we may introduce the
discrete times and study the Fourier transform $\hat{\mathcal{C}_{*}}(p_{0})$, where
\begin{eqnarray}
 \mathcal{C}_{*}(\tau=t-s)=\left \{
\begin{array}{c} e^{-(t-s)A^0}~C, ~~t\geq s,  \\
C~e^{-(s-t)A^{0^\dagger}}, ~~t\leq s.\end{array}\right .
\end{eqnarray}
Then, with the inverse Fourier transform of  $\hat{\mathcal{C}_{*}}^{-1}(p_{0})$, we obtain an expression (but huge and unclear)  for $\mathcal{D}^{-1}$.  Instead of that, we propose the
use of an approximated expression (but with the main part of $\mathcal{C}$), derived as described below.

To begin, we consider the regime of strong pinning. For strong pinning $M > \alpha^{2}$, we have $\cosh(\tau\rho) = \cos(\tau\tilde{\rho})$ and
$\sinh(\tau\rho)/\rho = \sin(\tau\tilde{\rho})/\tilde{\rho}$, where $\tilde{\rho}$ is given by $\tilde{\rho}=\left(M-\alpha^2\right)^{1/2}$. Taking the Fourier transform  of $\exp[-|\tau|\alpha]\cos(\tilde{\rho}\tau)$,  we have
\begin{equation*}
\widehat{\exp[-|\tau|\alpha]\cos(\tilde{\rho}\tau)} = 2\alpha\left\{ \frac{M + p_{0}}{(M + p_{0}^{2})^{2} - 4(M-\alpha^{2})p_{0}^{2}} \right\} \equiv \hat{D}(p_{0}) ,
\end{equation*}
where we used continue times above just for ease of computation; discrete times lead to a similar expression with $1-\cos(p_{0})$ replacing $p_{0}^{2}$.
The second part of $\exp[-|\tau| A^{0}]$, with $\sin(\tau\tilde{\rho})/\tilde{\rho}$, involves a matrix whose diagonal terms will be very small for large $\tilde{\rho}$. Due to $M$, the off diagonal terms are not small in principle, but they are related to the ``crossed'' part $qp$. These terms will be insignificant  later (they will be small) in the interacting model controlled by the cluster expansion:  after a scaling, $qp$ will involve a small factor $1/\lambda^{1/3}$ and it will be easily controlled in the case of large $\lambda$. Here, for simplicity, we ignore this off diagonal part. Hence, choosing e.g. $M = 3\alpha^{2}$, we have $\hat{D}^{-1}(p_{0}) \simeq \frac{M}{2\alpha} + \frac{c}{\alpha} p_{0}^{2}$, where $c$ is a numerical factor.

In short,  by assuming the approximation above (discarding the off diagonal terms in the original Gaussian covariance), and still taking discrete times in a lattice with spacement $\varepsilon$,  we can  write the Gaussian measure
as proposed in Eq.(\ref{quadratica}).

\end{document}